\documentclass[12pt]{elsarticle}

\usepackage{amsmath,amssymb,amsthm}
\usepackage[a4paper]{geometry}
\usepackage[utf8]{inputenc}
\usepackage[T1]{fontenc}
\usepackage{lmodern}
\usepackage{mathrsfs}
\usepackage{microtype}
\usepackage{latexsym}
\usepackage[colorlinks]{hyperref}
\usepackage{cleveref}
\usepackage{graphicx}
\usepackage{bookmark}
\usepackage{booktabs}
\usepackage{enumitem}
\usepackage{diagbox}
\usepackage{lineno}
\usepackage{float}
\usepackage{color}
\usepackage{bm}
\usepackage{longtable}

\newtheorem{thm}{Theorem}
\newtheorem{lem}{Lemma}
\newtheorem{prop}{Proposition}

\newtheorem{coro}{Corollary}
\theoremstyle{definition}

\newtheorem{ex}{Example}
\newtheorem{rem}{Remark}

\renewcommand{\vec}[1]{\boldsymbol{#1}}

\begin{document}
	
\begin{frontmatter}

\journal{arXiv}



\title{New bounds for codes over Gaussian integers based on the Mannheim distance\tnoteref{fund}
}
\tnotetext[fund]{The work of Minjia Shi was supported by the National Natural Science Foundation of China under Grant 12471490. 
Jon-Lark Kim was supported in part by the BK21 FOUR (Fostering Outstanding Universities for Research) funded by the Ministry of Education(MOE, Korea) and National Research Foundation of Korea(NRF) under Grant No. 4120240415042 and by Basic Science Research Program through the National Research Foundation of Korea(NRF) funded by the Ministry of Science and ICT under Grant No. RS-2025-24534992.
}


\author[ahu1]{Minjia Shi\corref{CorAuthor}}
\cortext[CorAuthor]{Corresponding author.	E-mail address: smjwcl.good@163.com (M. Shi).}
\author[ahu1]{Xuan Wang}
\author[sog1,sog2]{Junmin An}
\author[sog1,sog2]{Jon-Lark Kim}

\affiliation[ahu1]{organization={Key Laboratory of Intelligent Computing Signal Processing, Ministry of Education},
            addressline={School of Mathematical Sciences, Anhui University},
            city={Hefei},
            postcode={230601},
            state={Anhui},
            country={China}}
\affiliation[sog1]{organization={Department of Mathematics, Sogang University},
            addressline={35, Baekbeom-ro},
            city={Mapo-gu},
            postcode={04107},
            state={Seoul},
            country={Republic of Korea}}
\affiliation[sog2]{organization={Institute for Mathematical and Data Sciences, Sogang University},
            addressline={35, Baekbeom-ro},
            city={Mapo-gu},
            postcode={04107},
            state={Seoul},
            country={Republic of Korea}}

\begin{abstract}
We study linear codes over Gaussian integers equipped with the Mannheim distance. We develop Mannheim-metric analogues of several classical bounds. We derive an explicit formula for the volume of Mannheim balls, which yields a sphere packing bound and constraints on the parameters of two-error-correcting perfect codes. We prove several other useful bounds, and exhibit families of codes meeting these bounds for some parameters, thereby showing that these bounds are tight. We also discuss self-dual codes over Gaussian integers and obtain upper bounds on their minimum Mannheim distance for certain parameter regions using a Mannheim version of the Macwilliams-type identity. Finally, we present decoding algorithms for codes over Gaussian integer residue rings. We give examples showing that certain errors which are not correctable under the Hamming metric become correctable under the Mannheim metric.
\end{abstract}




\end{frontmatter}




\section{Introduction}
Classical coding theory has been developed primarily for codes over finite fields under the Hamming distance. A significant turning point in this framework was the observation in~\cite{IEEE-HKCSS-Z4} that certain optimal nonlinear binary codes are connected to linear $\mathbb{Z}_4$ codes under the Lee distance via the Gray map, which provides the precise correspondence. This connection has stimulated substantial interest in the study of codes over various rings equipped with the Lee distance~\cite{IEEE-C-Z2k, PIT-CH, IEEE-GS-chainring}. Moreover, Chiang and Wolf~\cite{IC-CW-leechannel} showed that, for certain channel models, codes with the Lee distance are more suitable than those with the Hamming distance. Since then, a series of works have investigated their relevance to a variety of practical applications including phase shift keying (PSK)~\cite{IEEE-EY-Leemulti,IEEE-RS-LeeBCH,IEEE-S-LeeRM}. However, it is known that the codes under the Hamming or Lee distance are not well suited for error correction for quadrature amplitude modulation (QAM) signals, which is widely used in communication systems such as HSDPA and LTE, as well as in digital TV and Wi-Fi as discussed in~\cite{FCN-SLLK-QAM}. This is because QAM constellations lie in two-dimensional signal space, in which both the Hamming and Lee distances are inappropriate. To address this problem, Huber~\cite{IEEE-H-gaussian} proposed codes over the Gaussian integers equipped with the Mannheim distance and showed that such codes can be effectively applied to QAM signal constellations. His work has initiated a line of research on codes over Gaussian integers with the Mannheim distance.

In~\cite{AAECC-H-MTTDMM}, Huber established a MacWilliams-type identity for Mannheim weight enumerators by representing finite fields using Gaussian integers. Martinez et al.~\cite{TIT-MBG-metrics} developed a new metric for codes over Gaussian integers induced by circulant graphs and characterized perfect codes with respect to their metric. Bouyuklieva~\cite{B-Proc} provided an overview of coding-theoretic applications of the Gaussian integers, and discussed one Mannheim error–correcting (OMEC) codes over residue-class rings together with their decoding. Ozen and Güzeltépe~\cite{QIC-OG-quantum} constructed quantum codes from codes over Gaussian integers equipped with the Mannheim metric, and presented examples whose parameters are better than or comparable to previously known quantum codes. Matsui~\cite{IEEE-M-CM} proposed an algorithm to construct and search for self-orthogonal and self-dual codes over Gaussian integers using the Chinese Remainder Theorem. Sajjad et al.~\cite{IEEE-SSAA-GaussianBCH} provided a construction of BCH codes over Gaussian integers and presented a decoding algorithm for those codes based on modified Berlekamp-Massey algorithm. Although a wide range of studies has been conducted, bounds on the minimum distance of codes under the Mannheim
metric, and optimal codes meeting such bounds have been less studied. This motivates us to establish
Mannheim metric analogues of classical bounds and to investigate the existence of codes meeting such bounds.

The sphere packing bound is one of the most fundamental results in coding theory. If a code can correct up to $t$ errors, then its size is bounded above by the size of the ambient space divided by the volume of a ball of radius $t$ in the underlying metric. The codes meeting the sphere packing bound are called perfect codes, and the existence and classification of perfect codes have been extensively studied over various fields and rings under a wide range of distance metrics~\cite{IEEE-BR-1perfect, SAM-GW-Leeperfect, SAM-T-fieldperfect}.

 In this paper, we compute the volume of Mannheim balls, and we derive a Mannheim metric version of the sphere packing bound. As a consequence, we obtain a necessary condition for the existence of a 2-error-correcting perfect code over the Gaussian integer residue rings. In particular, our results imply that if a 2-error-correcting perfect code under the Mannheim metric exists, then the smallest field over which such a code can exist is $\mathbb{F}_{29}$ and the smallest candidate parameter is $[n, k, d_\pi]=[10, 8, 5]$, where $d_\pi$ is the minimum Mannheim distance of the code.

In addition to determining optimal parameters for general linear codes, the study of optimal self-dual codes is also a central topic because of their importance in coding theory. Many bounds for the minimum distances of self-dual codes for various parameters have been provided and also various construction methods have been developed to obtain self-dual codes meeting such bounds~\cite{FFA-H-self-dual,FFA-H-self-dual-2,JCTA-LL-self-dual,IEEE-K-building,IEEE-KC-building,DCC-CK-self-dual,IEEE-WLZ}. Using a Mannheim metric version of the MacWilliams-type identity by Huber, we obtain the upper bounds on the minimum distances of self-dual codes over $\mathbb{F}_{13}$ and $\mathbb{F}_{17}$ for several lengths.

Our main contributions are as follows. We consider codes over $\mathbb{Z}[i]/(\pi)$ where $(\pi)$ is a maximal ideal of $\mathbb{Z}[i]$. If $\pi\equiv 3\pmod 4$, then the Mannheim weight distribution of a code coincides with the Lee weight distribution of the code induced from it. For the remaining case when $\pi$ is a factor of a prime $p\equiv 1\pmod 4$, we derive Mannheim metric analogues of the sphere packing bound. In addition, we obtain several further bounds on code parameters in the Mannheim metric, including bounds for self-dual codes. Thus, we obtain upper bounds on the best achievable parameters for general linear codes and self-dual codes under the Mannheim distance. Moreover, we also present a decoding algorithm for codes equipped with the Mannheim distance. Since there are cases in which an error is not correctable under the Hamming distance but become correctable under the Mannheim distance, this further highlights the usefulness of our study.

This paper is organized as follows. Section 2 introduces preliminary concepts for codes over Gaussian integer residue rings under the Mannheim distance. In Section 3, we discuss the relations between Mannheim metric and other metrics such as Hamming metric and Lee metric. Section 4  establishes a Mannheim-metric version of the sphere-packing bound and discuss perfect codes with respect to the Mannheim distance. In Section 5, we derive bounds on the minimum Mannheim distance of self-dual codes and present optimal self-dual codes meeting these bounds. In Section 6, we develop a decoding algorithm under the Mannheim distance. We conclude our paper in Section 7.

\section{Preliminaries}

\subsection{Gaussian integers}

The ring of Gaussian integers, denoted by $\mathcal{G}=\mathbb{Z}[i]$, is the set of complex numbers with integer real and imaginary parts, that is,
\[ \mathcal{G} = \{ x + y i \mid x, y \in \mathbb{Z} \}. \]
It is known that $\mathcal{G}$ forms a Euclidean domain under the norm $\mathcal{N}(x+yi)=x^2+y^2$. For any $\pi\in\mathcal{G}$, we define the quotient ring $\mathcal{G}_\pi=\mathcal{G}/(\pi)$, where $(\pi)$ is the principal ideal generated by $\pi$. When $\pi=a+bi$ satisfies $\gcd(a, b)=1$, then the map $\eta:\mathbb{Z}_{a^2+b^2}\to \mathcal{G}_\pi$ given by $\eta(g)=g\pmod\pi$ is a ring isomorphism, which gives $\mathcal{G}_\pi\cong\mathbb{Z}_{a^2+b^2}$.

According to \cite{TIT-MBG-metrics}, two quotient rings $\mathcal{G}_\alpha$ and $\mathcal{G}_\beta$ are isomorphic if and only if $\alpha = u \beta$ or $\bar{\alpha} = u \beta$ for some unit $u\in\{\pm 1, \pm i\}$. To choose a unique representative from each isomorphism class of quotient rings, we assume that $\pi=a+bi$ satisfies $0<a<b$ and $\gcd(a, b)=1$. Since $\mathcal{G}$ is a principal ideal domain, the quotient ring is a finite field is and only if $\pi$ is a Gaussian prime. These primes and their corresponding fields are classified as follows:
\begin{enumerate}[label=$\mathrm{(\arabic*)}$]
	\item[(i)] The prime $1 + i$. In this case, the quotient ring is isomorphic to the binary field, $\mathcal{G}/(1+i)\cong \mathbb{F}_2$.
	\item[(ii)] Rational primes $p \equiv 3 \pmod{4}$. The resulting quotient ring is $\mathcal{G}/(p)\cong\mathbb{F}_{p^2}$.
	\item[(iii)] Factors of rational primes $p \equiv 1 \pmod{4}$. If $p=a^2+b^2$, then $\pi=a+bi$ is a Gaussian prime, and the quotient ring is $\mathcal{G}/(\pi)\cong \mathbb{F}_p$.
\end{enumerate}

The metric structure of $\mathcal{G}_\pi$ is given by the Mannheim weight. For an element $\alpha\in\mathcal{G}_\pi$, the {\it Mannheim weight} is defined as
\[
{\rm wt}_{\pi}(\alpha)=|a|+|b|,
\]
where $a+bi$ is a representative of the residue class $\alpha\pmod\pi$ that minimizes $|a|+|b|$. This definition naturally extends to vectors in $\mathcal{G}_\pi^n$. For $\vec{x}=(x_1, \ldots, x_n)\in\mathcal{G}_\pi^n$, the Mannheim weight is the sum of the weights of its components, that is,
\[
   {\rm wt}_{\pi}(\vec{x}) = \sum_{i=1}^{n} {\rm wt}_{\pi}(x_i).
\]
The \textit{Mannheim distance} between two vectors $\vec{x}, \vec{y}\in\mathcal{G}_\pi^n$ is then defined as
\[
d_\pi(\vec{x}, \vec{y})={\rm wt}_{\pi}(\vec{x}-\vec{y}),
\]
which has been shown to be a metric in \cite[Theorem 2]{TIT-MBG-metrics}.

\subsection{Codes over Gaussian integers}

A \textit{linear code} $\mathcal{C}$ of length $n$ over $\mathcal{G}_\pi$ is defined as a subspace of $\mathcal{G}_\pi^n$. Although these codes can be viewed as conventional linear codes over the finite fields $\mathbb{F}_p$ or $\mathbb{F}_{p^2}$ depending on the choice of $\pi$, we focus on their properties with respect to the Mannheim metric throughout the paper.

The \textit{(Euclidean) inner product} on $\mathcal{G}_\pi^n$ is defined as
\[
\langle \vec{x}, \vec{y}\rangle=\sum_{i=1}^nx_iy_i,
\]
for $\vec{x}=(x_1, \ldots, x_n), \vec{y}=(y_1, \ldots, y_n)\in \mathcal{G}_\pi^n$.
For a code $\mathcal{C}\subseteq\mathcal{G}_\pi^n$, the \textit{(Euclidean) dual} of $\mathcal{C}$ is
\[
\mathcal{C}^\perp=\{\vec{x}\in\mathcal{G}_\pi^n~|~\langle \vec{x}, \vec{c}\rangle=0\mbox{ for all }\vec{c}\in\mathcal{C}\}.
\]
A code $\mathcal{C}$ is called \textit{self-orthogonal} if $\mathcal{C}\subseteq\mathcal{C}^\perp$, and \textit{self-dual} if $\mathcal{C}=\mathcal{C}^\perp$.

We denote the \textit{minimum Mannheim distance} of $\mathcal{C}$ by $d_\pi(\mathcal{C})$, whereas $d_H(\mathcal{C})$ denotes the \textit{minimum Hamming distance} of $\mathcal{C}$. We define $d_\pi(n, k)$ and $d_H(n, k)$ to be the maximum of the minimum Mannheim and Hamming distances among all $[n, k]$ codes over $\mathcal{G}_\pi$, respectively. An $[n, k]$ code achieving $d_\pi(\mathcal{C})=d_\pi(n, k)$ is called a \textit{Mannheim optimal code}, and one achieving $d_H(\mathcal{C})=d_H(n, k)$ is called a \textit{Hamming optimal code}.

Similarly, we define $d_\pi^{SD}(n)$ and $d_H^{SD}(n)$ as the maximum of the minimum Mannheim and Hamming distances among all $[n,n/2]$ self-dual codes over $\mathcal{G}_\pi$, respectively. A self-dual code is referred to as a \textit{Mannheim optimal self-dual code} or a \textit{Hamming optimal self-dual code} if its minimum distance reaches these maximums.

Two linear codes $\mathcal{C}_1$ and $\mathcal{C}_2$ are said to be \textit{Mannheim-equivalent}, denoted by $\mathcal{C}_1 \cong \mathcal{C}_2$, if there exists a monomial matrix $M$ whose non-zero entries are in $\{\pm1, \pm i\}$ such  that $\mathcal{C}_2\cong\mathcal{C}_1M$.

For self-dual codes, we consider another notion of equivalence. Two codes $\mathcal{C}_1$ and $\mathcal{C}_2$ are called \textit{$(1, -1, 0)$-monomial equivalent} if there exists a monomial matrix $M$ whose non-zero entries are in $\{1, -1\}$ such that $\mathcal{C}'=\mathcal{C}M$. In the paper, we utilize $(1, -1, 0)$-monomial equivalence for self-dual codes over $\mathcal{G}_\pi$ instead of the Mannheim equivalence defined above. This is because the Mannheim equivalence does not necessarily preserves the self-duality of a code under the Euclidean inner product, as it allows multiplication by the units $\pm i$.

\section{Relations Between Mannheim Distance and Other Distances}

In this section, we will discuss the relations between several metrics.

\subsection{$p\equiv 3\pmod 4$}

Recall that the \textbf{Lee weight} of $x \in \mathbb{Z}_p$ is defined as ${\rm wt}_L(x) = \min \{ x, p-x \}$, where $0 \leqslant x < p$.
Similarly, the Lee weight of the vector $\vec{x} = (x_1, \dots, x_n) \in \mathbb{Z}_p^n$ is defined as ${\rm wt}_L(\vec{x}) = \sum_{i=1}^{n} {\rm wt}_L(x_i)$.
There is close relation between the Mannheim weight and Lee weight when $p \equiv 3 \pmod{4}$.
\begin{prop}
	Let $x = a+bi \in \mathcal{G}_p$, where $p \equiv 3 \pmod{4}$ is prime, and $0 \leqslant a,b < p$.
	Then ${\rm wt}_{\pi}(x) = \min \{ a, p-a \} + \min \{ b, p-b \} = {\rm wt}_L(a) + {\rm wt}_L(b)$.
\end{prop}
\begin{proof}
	Since $x \equiv a + b i \pmod{p}$, then $x = (a + k_1 p) + (b + k_2 p)$ for some integers $k_1, k_2$.
	According to the definition of Mannheim weight, we have
	\begin{align*}
		{\rm wt}_{\pi}(x)
		& = \min \{ |a + k_1 p| + |b + k_2 p| \} \\
		& = \min_{k_1 \in \mathbb{Z}} |a + k_1 p| + \min_{k_2 \in \mathbb{Z}} |b + k_2 p| \\
		& = \min \{ a, p-a \} + \min \{ b, p-b \} \\
		& = {\rm wt}_L(a) + {\rm wt}_L(b),
	\end{align*}
	which completes the proof.
\end{proof}

Let $\vec{r} = (r_1, \dots, r_n)$ be a vector over $\mathcal{G}_p$, where $p \equiv 3 \pmod{4}$ and $r_j = x_j + y_j i$ for $1 \leqslant j \leqslant n$.
Define the map $\varphi: \mathcal{G}_p^n \to \mathbb{Z}_p^{2n}$ as
\[ \varphi(\vec{r}) = (x_1, \dots, x_n, y_1, \dots, y_n) = (\vec{x}, \vec{y}), \]
where $\vec{x} = (x_1, \dots, x_n)$, and $\vec{y} = (y_1, \dots, y_n)$.

\begin{lem}
	Keep the notations above.
	If $\lambda = \mu + \upsilon i \in \mathcal{G}_p$ with $0 \leqslant \mu, \upsilon \leqslant p-1$,
	then $\varphi(\lambda \vec{r}) = \mu \varphi(\vec{r}) + \upsilon \varphi(i\vec{r})$.
	Moreover, ${\rm wt}_{\pi}(\vec{r}) = {\rm wt}_L(\varphi(\vec{r}))$.
\end{lem}
\begin{proof}
	Just prove for $n = 1$.
	Assume that $r = x+y i \in \mathcal{G}_p$.
	Then
	\[
	\mu \varphi(r) + \upsilon \varphi(i r) = \mu (x, y) + \upsilon (-y, x) = (x \mu - y \upsilon, y \mu + x \upsilon) = \varphi(\lambda r).
	\]
	Moreover,
	\[ {\rm wt}_{\pi}(r) = {\rm wt}_L(x) + {\rm wt}_L(y) = {\rm wt}_L(\varphi(r)), \]
	which completes the proof.
\end{proof}

This gives the following theorem.
\begin{thm}\label{Gaussian-Lee}
	Let $\mathcal{C}$ be a linear $[n,k]$ code over $\mathcal{G}_p$ with generator matrix $G$ whose $j$-th row is $\vec{r}_j$, $1 \leqslant j \leqslant k$, where $p \equiv 3 \pmod{4}$.
	Let $\mathcal{C}'$ be the code over $\mathbb{F}_p$ spanned by the $2k$ vectors $\{ \varphi(\vec{r}_1), \varphi(i\vec{r}_1), \dots, \varphi(\vec{r}_k), \varphi(i\vec{r}_k) \}$.
	Then $\mathcal{C}'$ is a linear $[2n,2k]$ code over $\mathbb{Z}_p$.
	Moreover, the Mannheim weight distribution of $\mathcal{C}$ is the same as the Lee weight distribution of $\mathcal{C}'$.
\end{thm}

\begin{rem}
	Assume that $\vec{r}_j = \vec{x}_j + i \vec{y}_j$.
	Then the generator matrix of $\mathcal{C}'$ is
	\[
	\begin{pmatrix}
		\vec{x}_1 & \vec{y}_1 \\
		\dots & \dots \\
		\vec{x}_k & \vec{y}_k \\
		-\vec{y}_1 & \vec{x}_1 \\
		\dots & \dots \\
		-\vec{y}_k & \vec{x}_k
	\end{pmatrix}
	=
	\begin{pmatrix}
		X & Y \\
		-Y & X
	\end{pmatrix}.
	\]
	Conversely, given a linear $[n,k]$ code $\mathcal{C}_1$ over $\mathbb{Z}_p$ with generator matrix $X$,
	there exists a $k \times n$ matrix $Y$ over $\mathbb{Z}_p$ such that $X + Y i$ is a generator matrix of a linear $[n,k]$ code $\mathcal{C}$ over $\mathcal{G}_{p}$.
\end{rem}

In the classification of Gaussian primes, the Mannheim distance of $\pi=1+i$ is identical to the Hamming distance. For rational primes $p\equiv 3\pmod 4$, Theorem~\ref{Gaussian-Lee} shows that the Mannheim weight distribution is determined by the Lee weight distribution over $\mathbb{Z}_p$. Consequently, our paper focuses on case when $\pi$ is a Gaussian prime such that $\pi$ is a factor of a rational prime $p\equiv 1\pmod 4$.

\subsection{$p\equiv1\pmod 4$}

Let $\gamma = \xi^{(p-1)/4}$, where $\xi$ is a primitive element of $\mathbb{F}_{p}^{\ast}$.
Then $\gamma^4 = 1$.
Considering that $\eta: \mathbb{Z}_{a^2+b^2} \to \mathcal{G}_{\pi}$ is isomorphism, we have $\eta(\gamma) = i$, i.e., ${\rm wt}_{\pi}(\gamma) = 1$.
Let $H = \{ \pm 1, \pm \gamma \}$.
Then all the elements in $H$ have Mannheim weight $1$.

In fact, we have
\begin{lem} \label{lemma-unit-weight-2}
	Let $\pi = a + b i$ be the Gaussian prime such that $0 < a < b$ and $\gcd(a,b) = 1$.
	Then for the given $\alpha \in \mathcal{G}_{\pi}^{\ast}$, we have ${\rm wt}_{\pi}(\alpha) = {\rm wt}_{\pi}(u \alpha)$,
	where $u \in \{\pm 1, \pm i\}$.
\end{lem}
\begin{proof}
	Let $\alpha = x+yi \in \mathcal{G}_{\pi}^{\ast}$ such that $|x| + |y|$ is minimum.
	If ${\rm wt}_{\pi}(\alpha) \neq {\rm wt}_{\pi}(i \alpha)$, then there exists $x_0, y_0$ such that $i \alpha = x_0 i - y_0$ and $|x_0| + |y_0|$ is minimum, i.e., $|x| + |y| > |x_0| + |y_0|$, which is a contradiction since $\alpha = (i^3) \cdot i \alpha = x_0 + y_0 i$ and $|x| + |y|$ is minimum.
\end{proof}

Note $H$ is a subgroup of $\mathbb{F}_{p}^{\ast}$.
Hence, we get the following coset decomposition
\begin{equation} \label{eq-coset-decomposition}
	\mathbb{F}_{p}^{\ast} = i_1 H \cup i_2 H \cup \dots \cup i_{(p-1)/4} H,
\end{equation}
where $1 = i_1 < i_2 < \dots < i_{(p-1)/4}$.
By \Cref{lemma-unit-weight-2}, we know that all the elements of the same coset have the same Mannheim weight.

\begin{thm}{\rm(\cite[Theorem 6]{TIT-MBG-metrics})}
	Let $\pi = a + b i$ be a Gaussian prime such that $0 < a < b$, $\gcd(a,b) = 1$ and $\mathcal{N}(\pi)=p$ for some prime $p \equiv 1 \pmod{4}$. Let $\mathcal{W}_j$ denote the number of elements in $G_\pi$ with Mannheim weight $j$. Then
	\[
	\mathcal{W}_j =
	\begin{cases}
		1, & ~\mathrm{if}~j = 0; \\
		4j, & ~\mathrm{if}~1 \leqslant j \leqslant t; \\
		4(b-j), & ~\mathrm{if}~t < j \leqslant b-1,
	\end{cases}
	\]
	where $t = (a+b-1)/2$. Moreover, the maximum Mannheim weight in $\mathcal{G}_\pi$ is $b-1$.
\end{thm}

\begin{prop}\label{prop-special-weight-2}
	Let $\pi = a+bi$ be the Gaussian prime, where $0 < a < b = a+1$ and $p = a^2 + (a+1)^2 \equiv 1 \pmod{4}$ is a prime.
	Then each element in $\mathbb{F}_{p}^{\ast}$ can be expressed as $x+y\gamma$, where $|x| + |y| \leqslant a$.
	Moreover, we have $wt_{\pi}(x+y\gamma) = m$ if $|x| + |y| = m \leqslant a$, where $\gamma = 2a+1$.
\end{prop}
\begin{proof}
	It is easy to find that
	\[ (2a+1)^4 = 16 a^4 + 32 a^3 + 24 a^2 + 82 + 1 = (8 a^2 + 8 a)(2 a^2 + 2a + 1) + 1, \]
	which implies that $(2a+1)^4 \equiv 1 \pmod{(2a^2+2a+1)}$, i.e., $\gamma^4 = 1$ and $\eta(\gamma) = i$.
	Hence, each element in $\mathbb{F}_{p}^{\ast}$ can be expressed as $x+y\gamma$.
	
	Since $b = a+1$, we have $t = (a+b-1)/2 = a$ and
	\[ \frac{p-1}{4} = \frac{a^2 + (a+1)^2 - 1}{4} = \frac{a(a+1)}{2} = 1 + 2 + \dots + a. \]
	According to the proof of \cite[Theorem 6]{TIT-MBG-metrics}, all the elements of $\mathbb{F}_{p}^{\ast}$ can expressed by $x+y \gamma$, where $|x| + |y| \leqslant t = a$.
	Hence, there are exactly $4j$ elements of Mannheim weight $j$ in $\mathcal{G}_{\pi}^{\ast}$ where $1 \leqslant j \leqslant a$.
	
	Let $\alpha = x + y \gamma$.
	Then \[ \alpha H = \{ x+y\gamma, -y+x \gamma, -x-y\gamma, y-x\gamma \}. \]
	That is to say, in each coset, there is exactly one element $x+y\gamma$ such that $x,y$ are both non-negative integers.
	Since there are exactly $(p-1)/4$ cosets and exactly $(p-1)/4$ elements $x_m + y_m \gamma$ with non-negative integers $x_m,y_m$ such that $x_m + y_m = m \leqslant a$,
	the element $x_m + y_m \gamma$ has Mannheim weight $m$ for each $1 \leqslant m \leqslant a$.
\end{proof}

\begin{rem}
	According to \Cref{prop-special-weight-2}, the set of the elements of weight $m \leqslant a$ is
	\[ S_m = mH \cup (m-1+\gamma)H \cup (m-2+2\gamma) \cup \dots \cup (m-k+k\gamma) \cup \dots \cup (1+(m-1)\gamma)H. \]
	Since $H$ is a subgroup of $\mathbb{Z}_p^{\ast}$, we have the coset decomposition:
	\[ \mathbb{Z}_p^{\ast} = H \cup S_2 H \cup S_3H \cup \dots \cup S_m H. \]
\end{rem}

\begin{coro}\label{coro-special-weight-2}
	Let $\pi = a+bi$ be the Gaussian prime, where $0 < a < b = a+1$ and $p = a^2 + (a+1)^2 \equiv 1 \pmod{4}$ is a prime.
	Let
	\[ \Gamma_m= \{ m a + (m+1), \dots, m a + j, \dots, \ (m+1)a \}, \quad m+1 \leqslant j \leqslant a, \quad 0 \leqslant m \leqslant a-1. \]
	Then $ma+j$ has Mannheim weight $j$ when $m$ is even and Mannheim weight $a+m+1-j$ when $m$ is odd.
\end{coro}

\section{General Bounds on Codes over Gaussian Integers when $p \equiv 1 \pmod{4}$}

Throughout this section, let $\pi = a+bi \in \mathcal{G}$, $0 < a < b$, $\gcd(a,b) = 1$ and $p = a^2+b^2 \equiv 1 \pmod{4}$.
In this section, we will discuss the values and bounds on $d_{\pi}(n,k)$, the maximum minimum Mannheim distance among all $[n,k]$ codes.

\begin{prop}
	Let $n$ and $k$ be positive integers such that $n \geqslant k$.
	Then $d_{\pi}(n,k) \leqslant d_{\pi}(d_H(n,k),1)$.
\end{prop}
\begin{proof}
	Let $\mathcal{C}$ be an $[n,k]$ code with minimum Hamming distance $d_H(n,k)$.
	Let $\vec{c} \in \mathcal{C}$ be the codeword of Hamming weight $d_H(n,k)$.
	Then $d_{\pi}(\mathcal{C}) \leqslant w_{\pi}(\vec{c})$.
	The result form that $\min \{ w_{\pi}(\lambda \vec{c}): 0 \neq \lambda \in \mathcal{G}_{\pi} \} \leqslant d_{\pi}(d_H(\mathcal{C}),1)$.
\end{proof}

\begin{prop}
	Let $\pi = a+bi \in \mathcal{G}$, $0 < a < b$, $\gcd(a,b) = 1$ and $p = a^2+b^2 \equiv 1 \pmod{4}$.
	Then we have
	\[ d_{\pi}(n_1+n_2,k) \geqslant d_{\pi}(n_1,k) + d_{\pi}(n_2,k), \]
	where $n_1,n_2$ are positive integers.
\end{prop}
\begin{proof}
	Let $G_{m,k}$ be the generator matrix of an $[m,k]$ code, and $G' = ( G_{n_1,k}, G_{n_2,k} )$.
	Then $d_{\pi}(n_1+n_2,k) \geqslant d_{\pi}(n_1,k) + d_{\pi}(n_2,k)$.
\end{proof}

\begin{coro}
	Let $\pi = a+bi \in \mathcal{G}$, $0 < a < b$, $\gcd(a,b) = 1$ and $p = a^2+b^2 \equiv 1 \pmod{4}$.
	If $n = qs + r$, where $q,s,r$ are non-negative integers, $q > 0$ and $0 \leqslant r < q$,
	then we have
	\[ s \cdot d_{\pi}(q,k) + d_{\pi}(r,k) \leqslant d_{\pi}(n,k), \]
	where $d_{\pi}(r,k) = 0$ if $r < k$.
\end{coro}
\begin{proof}
	Let $G_{m,k}$ be the generator matrix of an $[m,k]$ code, and
	\[ G' = ( G_{q,k}, \dots, G_{q,k}, G_{r,k} ). \]
	Then $s \cdot d_{\pi}(q,k) + d_{\pi}(r,k) \leqslant d_{\pi}(n,k)$.
\end{proof}

\subsection{Improved bounds}

Recall the following coset decomposition as in Equation \eqref{eq-coset-decomposition},
\[ \mathbb{F}_{p}^{\ast} = i_1 H \cup i_2 H \cup \dots \cup i_{(p-1)/4} H, \]
where $1 = i_1 < i_2 < \dots < i_{(p-1)/4}$.
By \Cref{lemma-unit-weight-2}, we know that all the elements of the same coset have the same Mannheim weight.
Moreover, we have the following.

\begin{prop} \label{prop-constant-weight}
	Let $\vec{w} = (i_1, i_2, \dots, i_{(p-1)/4})$ where the coordinates are from Equation \eqref{eq-coset-decomposition}.
	Then for each $\lambda \in \mathbb{F}_{p}^{\ast}$, we have
	\[ w_{\pi}(\lambda \vec{w}) = \sum_{j=1}^{(a+b-1)/2} j^2 + \sum_{j=(a+b-1)/2+1}^{b-1} j \cdot (b-j) = S(a,b). \]
\end{prop}
\begin{proof}
	Note that for each $\lambda \in \mathcal{G}_{\pi}^{\ast}$, the elements $\lambda i_1, \lambda i_2, \dots, \lambda i_{(p-1)/4}$ are still in different cosets.
	Thus, for each nonzero $\lambda$, the value of $w_{\pi}(\lambda \vec{w})$ is a constant.
	The result follows from \cite[Theorem 6]{TIT-MBG-metrics}.
\end{proof}

Next, for each $[n,k]$ code $\mathcal{C}$ with the generator matrix $G$, we can define the $[n (p-1)/4, k]$ code $\mathcal{L}(\mathcal{C})$ generated by
\[ \mathcal{L}(G) = (i_1 G, i_2 G, \dots, i_{(p-1)/4} G). \]

\begin{thm}\label{thm-weight-relation}
	Let $\mathcal{C}$ be an $[n,k]$ code over $\mathbb{F}_p$, where $p = a^2 + b^2 \equiv 1 \pmod{4}$.
	Then for each $\vec{v} \in \mathcal{C}$, we have  $w_{\pi}(\mathcal{L}(\vec{v})) = S(a,b) \cdot w_{H}(\vec{v})$.
	Moreover,
	\[ d_{\pi}(\mathcal{C}) \leqslant \frac{4 S(a,b)}{p-1} d_H(\mathcal{C}). \]
\end{thm}
\begin{proof}
	Let $\vec{v}' = \mathcal{L}(\vec{v}) = (i_1 \vec{v}, i_2 \vec{v}, \dots, i_{(p-1)/4} \vec{v}) \in \mathcal{C}'$ for some nonzero codeword $\vec{v} = (v_1, \dots, v_n) \in \mathcal{C}$.
	By Proposition \ref{prop-constant-weight}, for each nonzero $v_j$, we have
	\[ w_{\pi}((i_1 v_j, i_2 v_j, \dots, i_{(p-1)/4} v_j)) = S(a,b), \]
	since the elements $i_1 v_j, i_2 v_j, \dots, i_{(p-1)/4} v_j$ are still in different cosets.
	Then
	\[ w_{\pi}(\vec{v}') = \sum_{j=1}^{(p-1)/4} w_{\pi}(i_j \vec{v}) = S(a,b) \cdot w_{H}(\vec{v}). \]
	Then
	\begin{align*}
		d_{\pi}(\mathcal{C})
		& = \min_{\vec{0} \neq \vec{v} \in \mathcal{C}} \min_{1 \leqslant j \leqslant (p-1)/4} w_{\pi}(i_j \vec{v}) \leqslant \min_{\vec{0} \neq \vec{v} \in \mathcal{C}} \frac{4 w_{\pi}(\mathcal{L}(\vec{v}))}{p-1} = \min_{\vec{0} \neq \vec{v} \in \mathcal{C}} \frac{4 S(a,b) \cdot w_{H}(\vec{v})}{p-1} \\
		& = \frac{4 S(a,b)}{p-1} \cdot d_{H}(\mathcal{C}),
	\end{align*}
	which completes the proof.
\end{proof}

\begin{coro}
	Let $n \geqslant k$ and $p = a^2 + b^2 \equiv 1 \pmod{4}$.
	Then
	\begin{equation}\label{eq-upper-bound}
		d_{\pi}(n,k) \leqslant \frac{4 S(a,b)}{p-1} d_H(n,k).
	\end{equation}
\end{coro}

The above upper bound is tight for some special parameters.

\begin{ex}
	Let $\mathcal{S}_{k,p}$ be the Simplex code of length $(p^k-1)/(p-1)$, dimension $k$ and minimum distance $p^{k-1}$ over $\mathbb{F}_p$.
	Let $\mathcal{S}'_{k,p} = \{ (\vec{v}, 2\vec{v}, 3 \vec{v}, \dots, (p-1) \vec{v}) : \vec{v} \in \mathcal{S}_{k,p} \}$.
	Then $\mathcal{S}'_{k,p}$ has length $p^k-1$ and minimum distance $(p-1) p^{k-1}$.
	By Theorem \ref{thm-weight-relation}, for each nonzero $\vec{v} \in \mathcal{S}'_{k,p}$, we have
	\[ w_{\pi}(\vec{v}) = 4S(a,b) \cdot p^{k-1} = \frac{4 S(a,b)}{p-1} d_{H}(\vec{v}). \]
	Moreover, it is known that $\mathcal{S}'_{k,p}$ meets the Griesmer bound, since
	\[ p^k-1 = \sum_{i=0}^{k-1} \left\lceil \frac{d}{p^i} \right\rceil = \sum_{i=0}^{k-1} \left\lceil \frac{(p-1)p^{k-1}}{p^i} \right\rceil = (p-1)(1 + p + \dots + p^{k-1}). \]
	That is to say, $d_H(p^k-1, k) = (p-1) p^{k-1}$ and $d_{\pi}(p^k-1,k) = 4 S(a,b) p^{k-1}$.
\end{ex}

\begin{coro}
	Let $\pi = a+bi \in \mathcal{G}$, $0 < a < b$, $\gcd(a,b) = 1$ and $p = a^2+b^2 \equiv 1 \pmod{4}$.
	Assume that $p' = (p-1)/4$ and $n = sp'+t$, where $n,s,t$ are non-negative integers and $t<p'$.
	If $s \geqslant k$, then
	\begin{equation} \label{eq-lower-bound}
		d_{\pi}(n,k) \geqslant S(a,b) \cdot d_H(s,k) + d_{\pi}(t,k),
	\end{equation}
	where we assume $d_{\pi}(t,k) = 0$ if $t < k$.
	Moreover, we have
	\[ S(a,b) d_H\left(\left\lfloor \frac{4n}{p-1} \right\rfloor, k\right) \leqslant d_{\pi}(n,k) \leqslant \frac{4S(a,b)}{p-1} d_{H}(n,k). \]
\end{coro}

\subsection{Special values on $k$ and $q$}

If $k = 1$, then $d_{H}(n,1) = n$ and we have
\[ \left\lfloor \frac{4 n}{p-1} \right\rfloor \cdot S(a,b) \leqslant d_{\pi}(n,1) \leqslant \left\lfloor \frac{4 n S(a,b)}{p-1} \right\rfloor. \]
It is clear that when $(p-1) \mid 4n$, we have
\[ d_{\pi}(n,1) = \frac{4 n S(a,b)}{p-1}. \]

\begin{prop}
	Let $\pi = 2 + 3 i$.
	Then
	\begin{equation} \label{eq-upper-bound-F13-1}
		d_{\pi}(n,1) = 2n - \left\lceil \frac{n}{3} \right\rceil.
	\end{equation}
\end{prop}
\begin{proof}
	It is easy to check that $S(2,3) = 5$ and $(13-1)/4 = 3$.
	Assume that $n = 3s + t$, where $0 \leqslant t < 3$.
	So the bound \eqref{eq-upper-bound-F13-1} is true when $t = 0$.
	For otherwise, by the bounds \eqref{eq-upper-bound} and \eqref{eq-lower-bound}, we have
	\[
	5s + d_{\pi}(t,1) \leqslant d_{\pi}(n,1) \leqslant \left\lfloor \frac{5n}{3} \right\rfloor = 2n - \left\lceil \frac{n}{3} \right\rceil = 5s+2t-1 =
	\begin{cases}
		5s+1, & \text{if} \ t=1, \\
		5s+3, & \text{if} \ t=2.
	\end{cases}
	\]
	The result follows from $d_{\pi}(1,1) = 1$ and $d_{\pi}(2,1) = 3$.
\end{proof}

\begin{coro}
	Let $\pi = 2+3i \in \mathcal{G}$.
	Then
	\[ d_{\pi}(n,k) \leqslant 2(n-k+1) - \left\lceil \frac{n-k+1}{3} \right\rceil. \]
\end{coro}
\begin{proof}
	The result just follows from that $d_{\pi}(n,k) \leqslant d_{\pi}(n-k+1, 1)$.
\end{proof}

\begin{ex}\label{ex-Z13}
	The above bound is tight.
	We have the following examples.
	\begin{enumerate}
		\item[$(1)$] The case $k=1$ is trivial.
		\item[$(2)$] The codes generated by
		\[
		G_{3,2} =
		\begin{pmatrix}
			1 & 0 & 2 \\
			0 & 1 & 4
		\end{pmatrix},
		G_{4,2} =
		\begin{pmatrix}
			1 & 0 & 2 & 4 \\
			0 & 1 & 4 & 2
		\end{pmatrix}
		\]
		have the minimum Mannheim distance $3$ and $5$, respectively.
	\item[$(3)$] When $k=3$, there are no linear $[5,3]$ codes meeting the bound.

	\end{enumerate}
\end{ex}

\begin{prop}
	Let $\pi = 1 + 4 i$ and $n = 4s+t$, where $0 < t < 4$.
	Then
	\begin{equation} \label{eq-bound-F17-1}
		8s+2t-1 \leqslant d_{\pi}(n,1) \leqslant 8s+2t.
	\end{equation}
\end{prop}
\begin{proof}
	It is easy to check that $S(1,4) = 8$ and $(17-1)/4 = 4$.
	By the bounds \eqref{eq-upper-bound} and \eqref{eq-lower-bound}, we have
	\[
	8s + d_{\pi}(t,1) \leqslant d_{\pi}(n,1) \leqslant 2n = 8s+2t.
	\]
	The result follows from $d_{\pi}(1,1) = 1$, $d_{\pi}(2,1) = 3$, $d_{\pi}(3,1) = 5$.
\end{proof}

\begin{ex}
	Since $S(2,3) = 5$ and $d_{2+3i}(2,1) = 3$ by Example \ref{ex-Z13}, we have $d_{\pi}(n,1) = s \cdot S(a,b) + d_{\pi}(t,1)$,  when $\pi = 2+3i$.
\end{ex}

\begin{ex}
	If $p = 17$ and $\pi = 1+4i$, then $d_{\pi}(2,1) = 3$, $d_{\pi}(3,1) = 5$, $d_{\pi}(4,1) = 8$.
	Moreover, $S(1,4) = 8$.
\end{ex}

Next, we will consider the values on $d_{\pi}(n,n-1)$.

\begin{thm}
	Let $\pi = a+bi \in \mathcal{G}$, $0 < a < b$, $\gcd(a,b) = 1$ and $p = a^2+b^2 \equiv 1 \pmod{4}$.
	Then for each positive integer $n$, we have
	\[ d_{\pi}(n,n-1) =
	\begin{cases}
		2, & n > (p-1)/4, \\
		3, & n = (p-1)/4, \\
		d \geqslant 3, & n < (p-1)/4.
	\end{cases}
	\]
\end{thm}
\begin{proof}
	Let $\mathcal{C}$ be an $[n, n-1]$ code over $\mathcal{G}_{\pi}$ with the parity-check matrix $\vec{v} = (v_1, \dots, v_n)$.
	It $v_j = 0$ for some $j$, then $\vec{e}_j = (e_{j,1}, e_{j,2}, \dots, e_{j,n}) \in \mathcal{C}$ and $d_{\pi}(\mathcal{C}) = 1$, where $e_{j,k} = \delta_{j,k}$ and $\delta_{j,k}$ is the Kronecker symbol.
	Otherwise, all the coordinates of $\vec{v}$ are nonzero.
	If $n > (p-1)/4$, then there exist $1 \leqslant j < k \leqslant n$ such that $v_j, v_k$ belong to the same coset.
	Thus $\vec{e}_j - v_j v_k^{-1} \vec{e}_k \in \mathcal{C}$ and $d_{\pi} = 2$ since $v_j v_k^{-1} \in \{ \pm 1, \pm i\}$.
\end{proof}

\begin{ex}
	Consider the field $\mathbb{Z}_{41} \cong \mathcal{G}_{\pi}$, where $\pi = 4+5i$.
	Note that $9^4 \equiv 1 \pmod{41}$, so $\{ 1,9,32,40 \} = H \leqslant \mathbb{Z}_{41}^{\ast}$ and we have the decomposition:
	\[ \mathbb{Z}_{41}^{\ast} = H \cup 2H \cup 3H \cup 4H \cup 6H \cup 7H \cup 8H \cup 11H \cup 12H \cup 16H. \]
	Moreover, the Mannheim weight of each coset leader is listed in Table \ref{tab-Z41}.
	Consider the code generated by the vector $(1,3)$.
	It is easy to check that its minimum  Mannheim distance is $4$.
	Furthermore, $d_{\pi}(2,1) = 4$.
	\begin{table}[htbp]
		\caption{Mannheim weight of the coset leaders in $\mathbb{Z}_{41}$}
		\label{tab-Z41}
		\begin{center}
				\begin{tabular}{|c|c|c|c|c|c|c|c|c|c|c|}
					\hline
					coset & 1 & 2 & 3 & 4 & 6 & 7 & 8 & 11 & 12 & 16 \\
					\hline
					$\mathcal{G}_{4+5i}$ & 1 & 2 & 3 & 4 & $-3+i$ & $-2+i$ & $-1+i$ & $1+2i$ & $1+3i$ & $-2-i$ \\
					\hline
					weight & 1 & 2 & 3 & 4 & 4 & 3 & 2 & 3 & 4 & 4 \\
					\hline
				\end{tabular}
		\end{center}
	\end{table}
\end{ex}

\begin{ex}
	Consider the field $\mathbb{Z}_{61} \cong \mathcal{G}_{\pi}$, where $\pi = 5+6i$.
	Note that $11^4 \equiv 1 \pmod{61}$, so $\{ 1,11,50,60 \} = H \leqslant \mathbb{Z}_{61}^{\ast}$ and we have the coset decomposition $\mathbb{Z}_{61}^{\ast} = \cup_{j \in I} j H$, where
	\[ I = \{ 1,2,3,4,5,7,8,9,10,13,14,15,19,20,25 \}. \]
	Moreover, the Mannheim weight of the coset leader are listed in Table \ref{tab-Z61}.
	Consider the code generated by the vector $(1,4)$.
	It is easy to check that its minimum Mannheim distance is $5$.
	Furthermore, $d_{\pi}(2,1) = 5$.
	\begin{table}[htbp]
		\caption{Mannheim weight of the coset leaders in $\mathbb{Z}_{61}$}
		\label{tab-Z61}
		\resizebox{\columnwidth}{!}{
			\begin{tabular}{|c|c|c|c|c|c|c|c|c|c|c|c|c|c|c|c|}
				\hline
				coset & 1 & 2 & 3 & 4 & 5 & 7 & 8 & 9 & 10 & 13  & 14 & 15 & 19 & 20 & 25   \\ \hline
				$\mathcal{G}_{5+6i}$ & 1 & 2 & 3 & 4 & 5 & -4+i & -3+i & -2+i & -1+i & 2+i & 3+i & 4+i & -3+2i & -2+2i & 3+2i \\ \hline
				weight & 1 & 2 & 3 & 4 & 5 & 5 & 4 & 3 & 2 & 3 & 4 & 5 & 5 & 4 & 5   \\ \hline
			\end{tabular}
		}
	\end{table}
\end{ex}

\section{Sphere Packing Bound} \label{sec-sphere-packing-bound}

Throughout this section, the weight refers to Mannheim weight, and $\pi = a+bi$ is a Gaussian prime such that $0 < a < b$, $\gcd(a,b) = 1$ and $p = a^2+b^2 \equiv 1 \pmod{4}$.
In this section, we will determine the Sphere Packing Bound with respect to Mannheim weight.

Let $V_{\pi}(s,n)$ be the number of vectors of length $n$ and weight $\leqslant s$.
Let $W_{\pi}(s,n)$ be the number of vectors of length $n$ and weight $s$.

According to  \cite[Theorem 6]{TIT-MBG-metrics}, it is known that the maximum weight of a nonzero element in $\mathbb{F}_p$ should be $b-1$.
For the vector $\vec{v} = (v_1, \dots, v_n) \in \mathbb{F}_p^n$, let $x_j$ be the number of its coordinates of weight $j$, $0 \leqslant j \leqslant b-1$.
Thus, $\vec{v}$ has weight $x_1 + 2 x_2 + \dots + (b-1) x_{b-1}$.
Let $\mathcal{N}_{\pi,s}$ be the set of all non-negative integer solution $(x_1, x_2, \dots, x_{b-1})$ of the equation
\[ x_1 + 2 x_2 + \dots + (b-1) x_{b-1} = s. \]

For each possible solution, there are exactly $\binom{n}{x_1 + \dots + x_{b-1}}$ choices of nonzero coordinates and there are exactly $\binom{x_1 + \dots + x_{b-1}}{x_1,\dots,x_{b-1}}$ arrangements for each above choice.
Hence
\[ W_{\pi}(s,n) = \sum_{(x_1,\dots,x_{b-1}) \in \mathcal{N}_{\pi,s}} \binom{n}{x_1 + \dots + x_{b-1}} \binom{x_1 + \dots + x_{b-1}}{x_1,\dots,x_{b-1}} \prod_{i=1}^{b-1} W_{\pi}(i,1)^{x_i}, \]
and
\[ V_{\pi}(s,n) = W_{\pi}(0,n) + W_{\pi}(1,n) + \dots + W_{\pi}(s-1,n) + W_{\pi}(s,n). \]

Similar to the classical sphere packing bound for Hamming distance, we have the following.
\begin{thm}
	Let $\mathcal{C}$ be a code of length $n$ over $\mathcal{G}_{\pi}$, and $d$ be its minimum Mannheim distance, where $\pi = a+bi \in \mathcal{G}$, $0 < a < b$, $\gcd(a,b) = 1$ and $p = a^2+b^2 \equiv 1 \pmod{4}$.
	Let $e = \lfloor (d-1)/2 \rfloor$. Then
	\[ |\mathcal{C}| \cdot V_{\pi}(e,n) \leqslant p^n.  \]
\end{thm}

\begin{ex}
	We have $W_{\pi}(1,n) = 4n$, and $W_{\pi}(2,n) = 8n + 16 \binom{n}{2} = 8n^2$.
\end{ex}

\begin{ex}
	The OMEC(one Mannheim error-correcting code)) code $\mathcal{C}$ over $\mathcal{G}_{\pi}$ has parameter $[n = (p^r-1)/4, n-r, d_{\pi} = 3]$, where $p \equiv 1 \pmod{4}$.
	It is known to be a perfect code since $|\mathcal{C}| \cdot W_{\pi}(1,n) = p^{n-r} \cdot (1+4n) = p^n$.
\end{ex}


\begin{lem}
	If there exists a $2$-perfect code $\mathcal{C}$ of length $n$ and dimension $n-r$ over $\mathcal{G}_{\pi}$, then $p^r = 8n^2 + 4n + 1$ and $r \in \{ 1,2,3,4,5 \}$.
\end{lem}
\begin{proof}
	If there exists a $2$-perfect code, then
	\[ V_{\pi}(2,n) \cdot |\mathcal{C}| = (8n^2 + 4n + 1) \cdot p^{n-r} = p^n, \]
	and the minimum Mannheim distance of $\mathcal{C}$ should be $5$ or $6$.
	Considering that $d_{\pi}(\mathcal{C}) \geqslant d_H(\mathcal{C})$ and $d_H(\mathcal{C}) \leqslant n-(n-r)+1 = r+1$, we have $r \leqslant 5$.
\end{proof}

\begin{thm}
	Let $1 \leqslant r \leqslant 5$.
	Then the equation $p^r = 8n^2 + 4n + 1$ has positive integer solutions if and only if $r = 1,2$.
	Moreover, when $r = 2$, the positive integer solutions satisfy $(n,p) = ((u_m + 2 v_m - 1)/4 , u_m + v_m)$, where $m$ is even and
	\begin{equation} \label{eq-solution}
		u_m = \frac{1}{2} \left( (3 + 2 \sqrt{2})^m + (3 - 2 \sqrt{2})^m \right), \quad
		v_m = \frac{1}{2\sqrt{2}} \left( (3 + 2 \sqrt{2})^m - (3 - 2 \sqrt{2})^m \right).
	\end{equation}
\end{thm}
\begin{proof}
	It is easy to find that $2p^r = (4n+1)^2 + 1$.
	\begin{enumerate}
		\item[$(1)$]
		If $r = 1$, then $p = 8n^2+4n+1$.
		
		\item[$(2)$]
		If $r = 2$, then
		\[ 8n^2+4n+1 = p^2 \Longleftrightarrow (2p)^2 - 2(4n+1)^2 = 2.  \]
		Let $x = 2p$ and $y = 4n+1$.
		Then we get a general Pell's equation $x^2 - 2y^2 = 2$.
		Consider the positive integer solution $(u_m,v_m)$ of the Pell's equation $u^2 - 2v^2 = 1$ whose fundamental solution is $(u_1,v_1) = (3,2)$, where $(u_m, v_m)$ are defined in Equation \eqref{eq-solution}.
		According to \cite[\S 4.9.3, pp. 98]{book-AA-QDE}, the equation $x^2 - 2y^2 = 2$ only has one fundamental solution $(x^*,y^*) = (2,1)$.
		By \cite[Theorem 4.1.3]{book-AA-QDE}, the solution of $x^2 - 2y^2 = 2$ is
		\[ x_m = x^* u_m + 2 y^* v_m = 2 u_m + 2 v_m = 2p, \quad y_m = y^* u_m + x^* v_m  = u_m + 2 v_m = 4n+1. \]
		Hence, \[ v_m = 4n+1-p, \quad u_m = p - v_m = 2p - 4n - 1. \]
		As for $u_m$ and $v_m$, we have
		\begin{align*}
			u_m
			& = \frac{1}{2} \left( (3 + 2 \sqrt{2})^m + (3 - 2 \sqrt{2})^m \right) \\
			& = \frac{1}{2} \sum_{i=0}^{m} \binom{m}{i} \left( 1+(-1)^i \right) 3^{m-i} (2\sqrt{2})^i \\
			& \equiv 3^m \equiv (-1)^m \pmod{4}, \\
			v_m
			& = \frac{1}{2\sqrt{2}} \left( (3 + 2 \sqrt{2})^m - (3 - 2 \sqrt{2})^m \right) \\
			& = \frac{1}{2\sqrt{2}} \sum_{i=0}^{m} \binom{m}{i} \left( 1-(-1)^i \right) 3^{m-i} (2\sqrt{2})^i \\
			& \equiv 2m \cdot 3^{m-1} \equiv (-1)^{m-1} \cdot \left( 1-(-1)^i \right) \pmod{4}.
		\end{align*}
		Thus, $8n^2+4n+1 = p^2$ has positive integer solutions if and only if $m$ is even since $4 \mid v_m$.
		
		\item[$(3)$]
		If $r = 3$,
		then \[ (4n+1)^2 + 1 = 2(8n^2+4n+1) = 2p^3 \implies (8n+2)^2 + 4 = (2p)^3.  \]
		Take $y = 8n+2$ and $x = 2p$.
		Then we have $y^2 = x^3 - 4$, whose positive integer solutions are $(x,y) = (2,2)$ and $(x, y) = (5,11)$ by \cite[Exercise 2.4.1]{book-ME-problem}.
		Thus, the equation $p^3 = 8n^2 + 4n + 1$ has no positive integer solutions.
		
		\item[$(4)$]
		If $r = 4$, then
		\[ (4n+1)^2 - 2p^4 = -1. \]
		By \cite{JLMS-L-QDE-4}, the only positive integer solutions of the equation $x^2 - 2y^4 = -1$ are $(x,y) = (1,1)$ and $(x,y) = (239,13)$.
		That is to say, $(n,p) = (0,1)$ or $(n,p) = (119/2,13)$.
		Hence, the equation $p^4 = 8n^2 + 4n + 1$ has no positive integer solutions.

		\item[$(4)$]
		If $r = 5$, then
		\[ (4n+1)^2 - 2p^5 = -1. \]
		Let $x=4n+1$ and $y=p$.
		Then over $\mathbb{Z}[i]$, we have \[ x^2 + 1 = 2y^5 \implies (x+i)(x-i) = (-i)(1+i)^2 y^5. \]
		Let $\delta = \gcd(x+i, x-i)$.
		Then $\delta \mid \gcd(2x,2i)$, which implies that $\delta \mid 2$.
		Note that $\gcd(2, x \pm i) = \gcd(2, 4n+1 \pm i) = \gcd(2,1\pm i) = 1 \pm i$.
		Thus, $\delta \neq 2$ and $\delta \neq 1$, i.e., $\delta = 1+i$.
		Then
		\[ \frac{x+i}{1+i} \cdot \frac{x-i}{1-i} = y^5, \quad \text{where}~ \gcd\left( \frac{x+i}{1+i}, \frac{x-i}{1-i} \right) = 1. \]
		Let $x+i = (1+i)(\ell_1 + \ell_2 i)^5$, where $\ell_1, \ell_2 \in \mathbb{Z}$.
		By computing, we get
		\begin{align*}
			(\ell_1 + \ell_2 i)^5 & = \ell_1^5 + 5 \ell_1^4 \ell_2 i + 10 \ell_1^3 \ell_2^2 i^2 + 10 \ell_1^2 \ell_2^3 i^3 + 5 \ell_1 \ell_2^4 i^4 + \ell_2^5 i^5 \\
			& = (\ell_1^5 - 10 \ell_1^3 \ell_2^2 + 5 \ell_1 \ell_2^4) + (5 \ell_1^4 \ell_2 - 10 \ell_1^2 \ell_2^3 + \ell_2^5) i,
		\end{align*}
		and
		\begin{align*}
			1 & = (\ell_1^5 - 10 \ell_1^3 \ell_2^2 + 5 \ell_1 \ell_2^4) + (5 \ell_1^4 \ell_2 - 10 \ell_1^2 \ell_2^3 + \ell_2^5) \\
			& = (\ell_1 + \ell_2)(\ell_1^4 + 4 \ell_1^3 \ell_2 - 14 \ell_1^2 \ell_2^2 + 4 \ell_1 \ell_2^3 + \ell_2^4) \\
			& = (\ell_1 + \ell_2)( (\ell_1 + \ell_2)^4 - 20 \ell_1^2 \ell_2^2),
		\end{align*}
		which implies that $\ell_1 + \ell_2 = 1$ and $20 \ell_1^2 \ell_2^2 = 0$, or $\ell_1 + \ell_2 = -1$ and $20 \ell_1^2 \ell_2^2 = 2$.
		Hence, $\ell_1 \ell_2 = 0$ and $\ell_1 = 1$, $\ell_2 = 0$, $x = 1$ or $\ell_1 = 0$, $\ell_2 = 1$, $x = -1$.
		And the only positive integer solution of the equation $x^2 + 1 = 2y^5$ is $(1, 1)$, while $4n + 1 > 1$.
		Thus, the equation $p^4 = 8n^2 + 4n + 1$ has no positive integer solutions.
	\end{enumerate}
	Therefore, the equation $p^r = 8n^2 + 4n + 1$ has positive integer solutions if and only if $r = 1,2$.
\end{proof}

\begin{rem}
	The smallest $2$  positive integer solutions of the equation $8n^2 + 4n + 1 = p^2$ that satisfied our conditions are $(n,p) = (10, 29)$ and $(11830, 33461)$.
	However, we do not know the existence of the $2$-perfect $[10, 8, d_{\pi} = 5]_{29}$ code, which is a possible perfect code with smallest length over the smallest field.
	As the field size becomes larger, maybe there exists a perfect code if there are infinite suitable primes produced by the solutions of $u^2 - 2v^2 = 1$.
\end{rem}

\begin{ex}
	Consider the field $\mathbb{F}_{29} \cong \mathcal{G}_{\pi}$, where $\pi = 2+5i$.
	Note that $12^4 \equiv 1 \pmod{29}$, then $\{ 12, 28, 17, 1 \} = \mathcal{F}_1 \leqslant \mathbb{F}_{29}^{\ast}$ and we have the decomposition:
	\[ \mathbb{F}_{29}^{\ast} = \mathcal{F}_1 \cup \mathcal{F}_2 \cup \mathcal{F}_3 \cup \mathcal{F}_4 \cup \mathcal{F}_6 \cup \mathcal{F}_8 \cup \mathcal{F}_{11}, \]
	where $\mathcal{F}_i = i \mathcal{F}_1$, $i=1,2,3,4,6,8,11$, and
	\begin{align*}
		\mathcal{F}_2 & = \{ 2,5,24,27 \}, \quad \mathcal{F}_3 = \{ 3,7,22,26\}, \quad \mathcal{F}_4 = \{ 4,19,10,25 \}, \\
		\mathcal{F}_{6} & = \{ 6,14,15,23 \}, \quad \mathcal{F}_{8} = \{ 8,9,20,21 \}, \quad \mathcal{F}_{11} = \{ 11,16,13,18 \} \\
	\end{align*}
	Moreover, we have
	\[ 4 = -1+2i, \quad 6 = 1+2i, \quad 8 = 1-3i, \quad 11 = -1-i. \]
	Thus, for the nonzero element in $\mathbb{F}_{29}$, $w_{\pi}(x) = 1$ if $x \in \mathcal{F}_1$, $2$ if $x \in \mathcal{F}_2 \cup \mathcal{F}_{11}$, $4$ if $x \in \mathcal{F}_8$ and $3$ for other $3$ cosets.
	The details are listed in Table \ref{tab-Z29}.
	\begin{table}[htbp]
		\caption{Mannheim weight of the elements in $\mathbb{Z}_{29}$}
		\label{tab-Z29}
		\begin{center}
			\resizebox{\columnwidth}{!}{
			\begin{tabular}{|c|c|c|c|c|c|c|c|c|c|c|c|c|c|c|}
				\hline
				$n$ & 1 & 2 & 3 & 4 & 5 & 6 & 7 & 8 & 9 & 10 & 11 & 12 & 13 & 14 \\
				\hline
				$\mathcal{G}_{2+5i}$ & 1 & 2 & 3 & $-1+2i$ & $2i$ & $2i+1$ & $-3i$ & $1-3i$ & $-3-i$ & $-2-i$ & $-1-i$ & $-i$ & $1-i$ & $2-i$ \\
				\hline
				weight & 1 & 2 & 3 & 3 & 2 & 3 & 3 & 4 & 4 & 3 & 2 & 1 & 2 & 3 \\
				\hline
			\end{tabular}}
		\end{center}
	\end{table}
\end{ex}

\section{Self-dual codes over Gaussian integers}

In this section, we investigate self-dual codes over finite fields viewed as Gaussian integer residue rings and study their minimum Mannheim distances. We consider self-dual codes defined with respect to the usual Euclidean inner product, in the same sense as conventional self-dual codes. However, our focus is on their minimum Mannheim distances while conventional self-dual codes are usually studied with respect to the Hamming distance.

\begin{thm}\label{thm-F13-4}
    Let $\mathcal{C}$ be a $[4, 2]$ self-dual code over $\mathbb{F}_{13}$. Then $d_\pi(\mathcal{C})\le 5$.
\end{thm}
\begin{proof}
For $\mathbb{F}_{13}\cong\mathcal{G}_{2+3i}$, we have the coset decomposition
\[
\mathbb{F}_{13}=\{0\}\cup \{1, 5, 8, 12\}\cup \{2, 10, 3, 11\}\cup \{4, 7, 6, 9\}
\]
where $\mbox{wt}_\pi(0)=0$, $\mbox{wt}_\pi(1)=1$, $\mbox{wt}_\pi(2)=2$ and $\mbox{wt}_\pi(4)=2$. Clearly, all elements in the same coset have the same Mannheim weight.

Suppose that $d_\pi(\mathcal{C})=6$. Since every nonzero element of $\mathbb{F}_{13}\cong \mathcal{G}_{2+3i}$ has Mannheim weight at most $2$, any vector of Hamming weight $1$ or $2$ has Mannheim weight at most $4$. Thus, every codeword in $\mathcal{C}$ has Hamming weight at least $3$. Since $d_H(\mathcal{C}) \le 3$ by the Singleton bound, it follows that $d_H(\mathcal{C})=3$. Hence, $\mathcal{C}$ contains a codeword $\mathbf{c}$ of Hamming weight $3$. Therefore, every nonzero coordinate of $\mathbf{c}$ must have Mannheim weight $2$, that is, every coordinate of $\mathbf{c}$ is in $A=\{2, 3, 4, 6, 7, 9, 10, 11\}\subset \mathbb{F}_{13}$. Since $A^2=\{a^2~|~a\in A\}=\{3, 4, 9, 10\}$ and no three elements of $A^2$ add to $0$, it follows that $\mathbf{c}\cdot\mathbf{c}\ne 0$, which is a contradiction. Therefore, $d_\pi(\mathcal{C})\le 5$.
\end{proof}

\begin{thm}\label{thm-F17-4}
    Let $\mathcal{C}$ be a $[4, 2]$ self-dual code over $\mathbb{F}_{17}$. Then $d_\pi(\mathcal{C})\le 5$.
\end{thm}
\begin{proof}
We first find the set $H$ of elements in $\mathbb{F}_{17}$ whose Mannheim weight is $1$. Since $4^2=-1$, $H=\{1, 4, 13, 16\}$. Then, the coset decomposition of $\mathbb{F}_{17}\cong\mathcal{G}_{1+4i}$ is given as follows:
\begin{align*}
\mathbb{F}_{17}&=\{0\}\cup H\cup 2H\cup 3H\cup 6H\\&=\{0\}\cup\{1, 4, 13, 16\}\cup\{3, 5, 12, 14\}\cup\{2, 8, 9, 15\}\cup \{6, 7, 10, 11\}.
\end{align*}
Clearly, $\mbox{wt}_\pi(0)=0$, $\mbox{wt}_\pi(1)=1$. Since $3\equiv -1+i$, $\mbox{wt}_\pi(3)=2$. Since there are exactly four elements in $\mathbb{Z}[i]/(1+4i)$ with Mannheim weight $2$, $\mbox{wt}_\pi(2), \mbox{wt}_\pi(6)>2$. Note that there is no element of Mannheim weight $3$ and four elements of Mannheim weight $4$. Thus, $\mbox{wt}_\pi(2)=4$ and $\mbox{wt}_\pi(6)=5$ since $6\equiv 2+i$.

By the Singleton bound, $d_H(\mathcal{C})\le 3$. First, assume that $d_H(\mathcal{C})=2$, and let $\mathbf{c}\in\mathcal{C}$ be a codeword of Hamming weight $2$. That is, $\mathbf{c}=(c_1, c_2, 0, 0)$ for some $c_1, c_2\in \mathbb{F}_{17}$. Since $\mathcal{C}$ is self-dual, $c_1^2=-c_2^2=16c_2^2$. Hence, $c_1=\pm 4c_2$. This shows that $c_1$ and $c_2$ are in the same coset of $H$. Choose $\alpha\in\mathbb{F}_{17}^\times$ such that $\alpha\cdot c_1\in 3H$. Multiplying $\mathbf{c}$ by $\alpha$, we obtain $\mbox{wt}_\pi(\alpha\cdot\mathbf{c})=4$. Then $d_\pi(\mathcal{C})\le 4 < 5$.

Next, assume that $d_H(\mathcal{C})=3$, and let $\mathbf{c}=(c_1,c_2,c_3,0)\in\mathcal{C}$ be a codeword of Hamming weight $3$. Since $c_1^{-1}\cdot(c_1, c_2, c_3, 0)=(1, c^{-1}_1c_2, c^{-1}_1c_3, 0)\in\mathcal{C}$, without loss of generality, we assume that $\mathbf{c}=(1, c_2, c_3, 0)$ for some $c_2, c_3\in \mathbb{F}_{17}$. Since $\mathcal{C}$ is self-dual, we have $1+c_2^2+c_3^2=0$. Let $A=\{x^2~:~x\in \mathbb{F}_{17}, x\ne 0\}=\{1, 2, 4, 8, 9, 13, 15, 16\}$. Then the only solutions for
\[
X+Y=-1\pmod{17}
\]
in $A$ are $(X, Y)=(1, 15)$ and $(X, Y)=(8, 8)$. We consider each case.
\begin{enumerate}
    \item[(i)] Let $c_2^2=1$ and $c_3^2=15$. Then $c_2\in H$, $c_3\in 6H$, and $\mbox{wt}_\pi(\mathbf{c})=1+1+5=7$.
    \item[(ii)] Let $c_2^2=c_3^2=8$. Then $c_2, c_3\in 3H$, and $\mbox{wt}_\pi(\mathbf{c})=1+2+2=5$.
\end{enumerate}
Suppose that $\mbox{wt}_\pi(\mathbf{c})=7$, that is, $\mathbf{c}=(1, c_2, c_3, 0)$ where $c_2\in H$ and $c_3\in 6H$. Then $\mbox{wt}_\pi(3\cdot\mathbf{c})=5$ since $3\cdot 1, 3\cdot c_2\in 3H$ and $3\cdot c_3\in H$. This implies that $\mathcal{C}$ must contain a codeword of Mannheim weight $5$. This completes the proof.
\end{proof}

\begin{thm}\label{thm-F17-6}
    Let $\mathcal{C}$ be a $[6, 3]$ self-dual code over $\mathbb{F}_{17}$. Then $d_\pi(\mathcal{C})\le 6$.
\end{thm}
\begin{proof}
    Note that $d_H(\mathcal{C})\le 4$ by the Singleton bound. If $d_H(\mathcal{C})<4$, as shown in the proof of Theorem~\ref{thm-F17-4}, $d_\pi(\mathcal{C})\le 5<6$.

    So, we assume that $d_H(\mathcal{C})=4$. Then, there exists a codeword $\mathbf{c}=(1, c_2, c_3, c_4, 0, 0)\in\mathcal{C}$ whose Hamming weight is $4$, and $c_2^2+c_3^2+c_4^2=-1$. By exhaustive search, we confirmed that all solutions of
    \[
    X+Y+Z=-1\pmod{17}
    \]
    in $A=\{1, 2, 4, 8, 9, 13, 15, 16\}$ are $(X, Y, Z)=(1, 16, 16),~(2, 15, 16),~(8, 9, 16)$. We consider each case.
    \begin{enumerate}
        \item[(i)] Let $c_2^2=1$ and $c_3^2=c_4^2=16$. Then $c_2, c_3, c_4\in H$, and $\mbox{wt}_\pi(\mathbf{c})=1+1+1+1=4$.
        \item[(ii)] Let $c_2^2=2$, $c_3^2=15$ and $c_4^2=16$. Then $c_2, c_3\in 6H$, $c_4\in H$ and $\mbox{wt}_\pi(\mathbf{c})=1+5+5+1=12$.
        \item[(iii)] Let $c_2^2=8$, $c_3^2=9$ and $c_4^2=16$. Then $c_2, c_3\in 3H$, $c_4\in H$ and $\mbox{wt}_\pi(\mathbf{c})=1+2+2+1=6$.
    \end{enumerate}
    Assume that $c_2, c_3\in 6H$ and $c_4\in H$. Then $3\cdot6H=H$ and $3\cdot H=3H$. Hence, $\mathrm{wt}_\pi(3\cdot\mathbf{c})=2+1+1+2=6$, and therefore $d_\pi(\mathcal{C})\le 6$.
\end{proof}

We present the best known values of $d_H^{SD}(n)$ and the upper bounds on $d_{\pi}^{SD}(n)$ in Table~\ref{tab-F13-general-bound} and~\ref{tab-F17-general-bound}, which are obtained from Theorem~\ref{thm-weight-relation} and~\cite[Table 4]{IEEE-KC-building}.

\begin{table}[H]
	\caption{The best known values of $d_H^{SD}(n)$ and upper bounds on $d_{\pi}^{SD}(n)$ over $\mathbb{F}_{13}$}
	\label{tab-F13-general-bound}
	\centering
	\begin{tabular}{c|c|c|c|c|c|c|c|c|c|c}
		\toprule
		length            & 2  & 4  & 6  & 8  & 10 & 12 & 14 & 16 & 18 & 20 \\ \hline
		$d_H^{SD}(n)$  & 2  & 3  & 4  & 5  & 6  & 6  & 8  & 8  & 8  & 10 \\ \hline
		$d_{\pi}^{SD}(n)$ & 3  & 5  & 6  & 8  & 10 & 10 & 13 & 13 & 13 & 16 \\ \toprule
		length            & 22 & 24 & 26 & 28 & 30 & 32 & 34 & 36 & 38 & 40 \\ \hline
		$d_H^{SD}(n)$  & 10 & 10 & 10 & 11 & 11 & 12 & 12 & 13 & 13 & 14 \\ \hline
		$d_{\pi}^{SD}(n)$ & 16 & 16 & 16 & 18 & 18 & 20 & 20 & 21 & 21 & 23 \\ \bottomrule
	\end{tabular}
\end{table}
\begin{table}[H]
	\caption{The best known values of $d_H^{SD}(n)$ and upper bounds on $d_{\pi}^{SD}(n)$ over $\mathbb{F}_{17}$}
	\label{tab-F17-general-bound}
	\centering
	\begin{tabular}{c|c|c|c|c|c|c|c|c|c|c}
		\toprule
		length            & 2  & 4  & 6  & 8  & 10 & 12 & 14 & 16 & 18 & 20 \\ \hline
		$d_H^{SD}(n)$  & 2  & 3  & 4  & 5  & 6  & 7  & 7  & 8  & 10 & 10 \\ \hline
		$d_{\pi}^{SD}(n)$ & 4  & 6  & 8  & 10 & 12 & 14 & 14 & 16 & 20 & 20 \\ \toprule
		length            & 22 & 24 & 26 & 28 & 30 & 32 & 34 & 36 & 38 & 40 \\ \hline
		$d_H^{SD}(n)$  & 10 & 10 & 10 & 11 & 12 & 12 & 12 & 13 & 14 & 14 \\ \hline
		$d_{\pi}^{SD}(n)$ & 20 & 20 & 20 & 22 & 24 & 24 & 24 & 26 & 28 & 28 \\ \bottomrule
	\end{tabular}
\end{table}

Let $\mathbb{F}_p=\{0\}\cup k_1H\cup\cdots\cup k_{(p-1)/4}H$ be the coset decomposition of $\mathbb{F}_p$. For a codeword $\mathbf{c}\in\mathcal{C}$, the \textit{composition} of $\mathbf{c}$ is defined as the sequence $t=(t_0, t_1, \ldots, t_{(p-1)/4})$ where $t_0$ is the number of components of $\mathbf{c}$ equal to $0$ and $t_j$ is the number of components of $\mathbf{c}$ in $k_jH$ for $1\le j\le (p-1)/4$. Let $A(t)$ be the number of codewords in $\mathcal{C}$ with composition $t=(t_0, t_1, \ldots, t_{(p-1)/4})$. Then the Gaussian integer enumerator polynomial for $\mathcal{C}$ is given as
\begin{align*}
\mbox{GWE}_\mathcal{C}(z_0, z_1, \ldots, z_{(p-1)/4})&=\sum_t A(t)z_0^{t_0}z_1^{t_1}\cdots z_{(p-1)/4}^{t_{(p-1)/4}} \\&=\sum_{\mathbf{c}\in \mathcal{C}}z_0^{t_0}z_1^{t_1}\cdots z_{(p-1)/4}^{t_{(p-1)/4}}.
\end{align*}
For an element $\gamma\in \mathbb{F}_p$, we define the character $\chi(\gamma)$ as $\chi(\gamma)=\xi^\gamma$ where $\xi$ is a primitive complex $p$-th root of unity, that is, $\xi=e^{2\pi i/p}$. The following theorem gives the Gaussian integer enumerator of $\mathcal{C}^\perp$.
\begin{thm}{\rm(\cite{AAECC-H-MTTDMM})}\label{Thm-Gaussian-enumerator}
    Let $\mathcal{C}$ be a code over $\mathcal{G}_\pi$ where $\pi = a+bi \in \mathcal{G}$, $0 < a < b$, $\gcd(a,b) = 1$ and $p = a^2+b^2 \equiv 1 \pmod{4}$. For $1\le j\le (p-1)/4$, choose a representative $\omega_j$ from $k_jH$ and set $\omega_0=0$. Then the Gaussian integer enumerator of the dual code $\mathcal{C}^\perp$ is
    \[
    {\rm{GWE}}_{\mathcal{C}^\perp}(z_0, z_1, \ldots, z_{(p-1)/4})=\frac{1}{|\mathcal{C}|}{\rm{GWE}}_\mathcal{C}(Z_0, Z_1, \ldots, Z_{(p-1)/4}),
    \]
    where
    \[
    Z_j=z_0+\sum_{s=1}^{(p-1)/4}(\chi(\omega_j\omega_s)+\chi(i\omega_j\omega_s)+\chi(-\omega_j\omega_s)+\chi(-i\omega_j\omega_s))z_s.
    \]
\end{thm}

\begin{lem}\label{lem-field-decomposition}
    Let $\mathbb{F}_{p}$ be the finite field of order $p$, where $p=a^2+b^2\equiv 1\pmod 4$ with $a<b$. Then, $\mathbb{F}_{p}$ has the coset decomposition given as follows:
    \[
    \mathbb{F}_{p}=\{0\}\cup \alpha H\cup \alpha^2H\cup \cdots\cup \alpha^{(p-1)/4}H
    \]
    where $\alpha$ is a primitive element of $\mathbb{F}_p$.
\end{lem}
\begin{proof}
    For some $1\le i\le j\le (p-1)/4$, let $\alpha^iH=\alpha^jH$, that is, $\alpha^{j-i}\in H$. Since there is a unique subgroup of $\mathbb{F}_p^*$ of order $4$, we have
    \[
    H=\langle \alpha^{(p-1)/4}\rangle=\{1, \alpha^{(p-1)/4}, \alpha^{2(p-1)/4}, \alpha^{3(p-1)/4}\}.
    \]
    Since $0\le j-i<(p-1)/4$, $j$ must be equal to $i$. Since the number of cosets of $H$ in $\mathbb{F}_p^*$ is $|\mathbb{F}_p^*|/|H|=(p-1)/4$, this completes the proof.
\end{proof}

Next, we provide a procedure to find an upper bound for the minimum Mannheim distance of self-dual codes over $\mathbb{F}_{p}$ where $p=a^2+b^2\equiv 1\pmod 4$ with $a<b$. For some $n\ge 2$, let
\begin{align*}
S=\{(s_0, s_1, s_2, \ldots, s_{(p-1)/4})\in\mathbb{Z}^{(p-1)/4+1}~|~s_0+s_1+s_2+\cdots+&s_{(p-1)/4}=n,\\&s_i\ge 0\mbox{ for all }i\}
\end{align*}
be the set of compositions of vectors in $\mathbb{F}_p^n$. Then
\[
|S|={n+(p-1)/4\choose (p-1)/4}=\frac{1}{((p-1)/4)!}\prod_{i=1}^{(p-1)/4}(n+i)
\]
and the Gaussian weight enumerator of a linear code $\mathcal{C}$ over $\mathbb{F}_p$ with length $n$ is given as
\[
{\rm GWE}_\mathcal{C}(z_0, z_1, z_2, \ldots, z_{(p-1)/4})=\sum_{t\in S}A(t)z_0^{t_0}z_1^{t_1}z_2^{t_2}\cdots z_{(p-1)/4}^{t_{(p-1)/4}}.
\]
Since it is the weight enumerator of a linear code $\mathcal{C}$, $A(t)\ge 0$ for every $t\in S$ and $A(n, 0, \ldots, 0)=1$. If $\mathcal{C}$ is self-dual, then we have
\[
\sum_{t\in S}A(t)=p^{n/2}.
\]
We denote
\[
\alpha_{j,s}
=
\chi(\omega_j\omega_s)+\chi(i\omega_j\omega_s)+\chi(-\omega_j\omega_s)+\chi(-i\omega_j\omega_s)
=
\sum_{a\in \omega_j\omega_s H}\xi^a
\]
for \(1\le j,s\le (p-1)/4\), where \(\omega_j\) is a representative of \(k_jH\).
Then
\[
Z_j=z_0+\sum_{s=1}^{(p-1)/4}\alpha_{j,s}z_s.
\]
By Theorem~\ref{Thm-Gaussian-enumerator}, the Gaussian weight enumerator of $\mathcal{C}^\perp$ is given as
\[
    {\rm{GWE}}_{\mathcal{C}^\perp}(z_0, z_1, z_2,\ldots, z_{(p-1)/4})=\frac{1}{p^{n/2}}{\rm{GWE}}_\mathcal{C}(Z_0, Z_1, Z_2,\ldots, Z_{(p-1)/4}),
\]
where
\[
\begin{pmatrix}
    Z_0\\Z_1\\Z_2\\\vdots\\Z_{(p-1)/4}
\end{pmatrix}=
\begin{pmatrix}
    1 & 4 & 4 & \cdots & 4\\
    1 & \alpha_1 & \alpha_2 & \cdots & \alpha_{(p-1)/4}\\
    1 & \alpha_2 & \alpha_3 & \cdots & \alpha_1\\
    \vdots & \vdots & \vdots & \cdots & \vdots\\
    1 & \alpha_{(p-1)/4} & \alpha_1 & \cdots & \alpha_{(p-1)/4-1}
\end{pmatrix}
\begin{pmatrix}
    z_0\\z_1\\z_2\\\vdots\\z_{(p-1)/4}
\end{pmatrix}.
\]
Note that ${\rm{GWE}}_{\mathcal{C}^\perp}(z_0,z_1,\ldots,z_{(p-1)/4})={\rm{GWE}}_{\mathcal{C}}(z_0,z_1,\ldots,z_{(p-1)/4})$ if $\mathcal{C}$ is self-dual. For each $n$, consider the following equation
\[
p^{\frac{n}{2}}\sum_{t\in S}A(t)z_0^{t_0}z_1^{t_1}z_2^{t_2}\cdots z_{(p-1)/4}^{t_{(p-1)/4}}=\sum_{t\in S}A(t)Z_0^{t_0}Z_1^{t_1}Z_2^{t_2}\cdots Z_{(p-1)/4}^{t_{(p-1)/4}}
\]
with $A(n, 0, \ldots, 0)=1$ and $A(t)\in\mathbb{Z}^{\ge 0}$ for every $t\in S$. For a given integer $d>0$, impose the constraints
\[
A(t)=0
\]
for every composition $t=(t_0,\ldots,t_{(p-1)/4})\in S$ satisfying
\[
0<\sum_{j=1}^{(p-1)/4} m_j t_j < d,
\]
where $m_j$ denotes the Mannheim weight of any element of the coset $k_jH$.
If the resulting system has a solution, but the analogous system with $d$ replaced by $d+1$ has no solution, then $d$ gives an upper bound on the minimum Mannheim distance of self-dual codes of length $n$ over $\mathbb{F}_p$.

To further reduce the number of variables in the equation, we introduce several additional strategies. Recall that
\[
\mathbb{F}_{p}\backslash\{0\}=k_1H\cup k_2H\cup\cdots\cup k_{(p-1)/4}H
\]
is the coset decomposition of $\mathbb{F}_p$. Define a bijection $\varphi:\mathbb{F}_{p}\to\mathbb{F}_{p}$ as $\varphi(x)=\alpha x$ for $x\in\mathbb{F}_{p}$ where $\alpha$ is a primitive element of $\mathbb{F}_p$. By Lemma~\ref{lem-field-decomposition}, the map $\varphi$ permutes the collection of cosets $\{k_jH\}_{j=1}^{(p-1)/4}$ in a single cycle. Let $\sigma$ be the permutation of the indices $\{1, 2, \ldots, (p-1)/4\}$ defined as $\varphi(k_jH)=k_{\sigma(j)}H$. Let $\mathbf{c}\in\mathcal{C}$ be a codeword with composition $t=(t_0, t_1, t_2,\ldots, t_{(p-1)/4})$, that is $\mathbf{c}$ consists of $t_0$ zeros and $t_j$ coordinates from $k_jH$ for $1\le j\le (p-1)/4$. Then $\varphi(\mathbf{c})=(\varphi(c_1), \varphi(c_2), \ldots, \varphi(c_n))$ is also a codeword of $\mathcal{C}$ and its composition $\varphi(t)$ is given by permuting the coordinates of $t$ according to $\sigma$. Thus, $A(t)=A(\varphi(t))$. By repeating this process, we have
\[
A(t)=A(\varphi(t))=A(\varphi^2(t))=\cdots=A(\varphi^{((p-1)/4)-1}(t)).
\]

Subsequently, we use the ideas from the proofs of Theorems~\ref{thm-F13-4},~\ref{thm-F17-4} and~\ref{thm-F17-6} to further restrict the values of $A(t)$. For any codeword $\mathbf{c}=(c_1, c_2, \ldots, c_n)$ of Hamming weight $h\le n$, since $c_1^{-1}\cdot\mathbf{c}\in\mathcal{C}$, we may assume that
\[
\mathbf{c}=(1, u_1, u_2, \ldots, u_{h-1}, 0, \ldots, 0).
\]
If $\mathcal{C}$ is self-dual, then since $\mathbf{c}\cdot\mathbf{c}=0$, $(u_1^2, u_{2}^2, \ldots, u_{h-1}^2)$ must be a solution of the following equation:
\begin{equation}\label{eq-LP-str2}
X_1+X_2+\ldots +X_{h-1}\equiv-1\pmod p.
\end{equation}
Since we are interested in the composition $t$ of a codeword $\mathbf{c}$, not the codeword $\mathbf{c}$ itself, we find all solutions to Equation~\eqref{eq-LP-str2} such that $X_i\le X_j$ for $1\le i< j\le h-1$ in the set $A=\{x^2~:~x\in\mathbb{F}_p^*\}$. Let $s=(s_1, s_2, \ldots, s_{h-1})$ be one solution. Note that for any $a\in A$, if $x^2=a$, then $(-x)^2=a$. Since the elements squaring to $a$ always belong to the same coset, any codeword $\mathbf{c}=(1, u_1, u_2, \ldots, u_{h-1}, 0,\ldots, 0)\in\mathcal{C}$ such that $u_j^2=s_j$ for $1\le j\le h-1$ has the same composition determined by $s$.

Let $t(s)$ denote the composition determined by the solution $s$ and let
\[
\mathcal{O}_s=\{t(s), \varphi(t(s)), \ldots, \varphi^{((p-1)/4)-1}(t(s))\}
\]
be the orbit of $t(s)$ under $\varphi$. Let $E_h$ be the set of solutions of Equation~\eqref{eq-LP-str2} with $X_i\le X_j$ for $1\le i< j\le h-1$ in the set $A=\{x^2~:~x\in\mathbb{F}_p^*\}$. Then, if $\mathbf{c}$ is a codeword of a self-dual code $\mathcal{C}$ with Hamming weight $h$, then its composition $t$ must be in $\bigcup_{s\in E_h}\mathcal{O}_s$. Thus, we set
\[
A(t)=0\quad\mbox{for}\quad t\in S\backslash\bigcup_{h=2}^{n}\bigcup_{s\in E_h}\mathcal{O}_s.
\]

To put it all together, we need to find the largest $d$ such that the equation
\begin{equation}\label{eq-LP-f}
p^{\frac{n}{2}}\sum_{t\in S}A(t)z_0^{t_0}z_1^{t_1}z_2^{t_2}\cdots z_{(p-1)/4}^{t_{(p-1)/4}}=\sum_{t\in S}A(t)Z_0^{t_0}Z_1^{t_1}Z_2^{t_2}\cdots Z_{(p-1)/4}^{t_{(p-1)/4}}
\end{equation}
has a solution under the constraints
\begin{enumerate}
    \item[(i)] $A(n, 0, \ldots, 0)=1$, $A(t)\in\mathbb{Z}^{\ge 0}$ for $t\in S$,
    \item[(ii)] $A(t)=0$ for $t\in S$ such that $0<\sum_{j=1}^{(p-1)/4} m_j t_j < d$,
    \item[(iii)] $A(t)=A(\varphi(t))=A(\varphi^2(t))=\cdots=A(\varphi^{((p-1)/4)-1}(t))$ for $t\in S$,
    \item[(iv)] $A(t)=0$ for $t\in S\backslash\bigcup_{h=2}^{n}\bigcup_{s\in E_h}\mathcal{O}_s$.
\end{enumerate}

From the above computation for $2\le n\le 14$, in $\mathbb{F}_{13}$, we obtain:
\[
\begin{array}{|c|c|c|c|c|c|c|c|}
\hline
    n   & 2 & 4 & 6 & 8 & 10 & 12 & 14 \\\hline
    d_\pi^*(n) & 2 & 5 & 5 & 7 & 9  & 10 & 12 \\
    \hline
\end{array}
\]
where $d_\pi^*(n)$ is the largest $d$ such that Equation~\eqref{eq-LP-f} has a solution. Also in $\mathbb{F}_{17}$, we obtain:
\[
\begin{array}{|c|c|c|c|c|c|c|}
\hline
    n & 2 & 4 & 6 & 8 & 10 & 12 \\\hline
    d_\pi^*(n) & 2 & 5 & 6 & 9 & 10  & 12 \\
    \hline
\end{array}
\]
These bounds improve those listed in Tables~\ref{tab-F13-general-bound} and~\ref{tab-F17-general-bound}.

\begin{lem}{\rm(\cite{DM-BGGHK-sdcode})}\label{sd-classify-13}
    There is a unique self-dual code of length $2$, and there are two self-dual codes of length $4$, five self-dual codes of length $6$, and $21$ self-dual codes of length $8$ over $\mathbb{F}_{13}$ up to $(1, -1, 0)$-equivalence.
\end{lem}

\begin{thm}\label{13-sd-opt8}
    Let $\pi=2+3i$. For $2\le n\le 8$, the values of $d_\pi^{SD}(n)$ are as follows:
    \[
    d_\pi^{SD}(n)=\begin{cases}
        2, & \mbox{if }n=2,\\
        5, & \mbox{if }n=4,\\
        5, & \mbox{if }n=6,\\
        6, & \mbox{if }n=8.
    \end{cases}
    \]
\end{thm}

\begin{lem}{\rm(\cite{DM-BGGHK-sdcode})}\label{sd-classify-17}
    There is a unique self-dual code of length $2$, and there are two self-dual codes of length $4$, six self-dual codes of length $6$, and $47$ self-dual codes of length $8$ over $\mathbb{F}_{17}$ up to $(1, -1, 0)$-equivalence.
\end{lem}

\begin{thm}\label{17-sd-opt8}
    Let $\pi=1+4i$. For $2\le n\le 8$, the values of $d_\pi^{SD}(n)$ are as follows:
    \[
    d_\pi^{SD}(n)=\begin{cases}
        2, & \mbox{if }n=2,\\
        5, & \mbox{if }n=4,\\
        6, & \mbox{if }n=6,\\
        7, & \mbox{if }n=8.
    \end{cases}
    \]
\end{thm}
Theorems~\ref{13-sd-opt8} and~\ref{17-sd-opt8} follow from direct computation using Lemma~\ref{sd-classify-13} and Lemma~\ref{sd-classify-17}. The computation results for each inequivalent code are given in Table~\ref{fig:code_results_13}. The notation for each code follows that used in~\cite{DM-BGGHK-sdcode}. Mannheim optimal self-dual codes are indicated by $(*)$.

\begin{center}
\begingroup
\small 
\setlength{\tabcolsep}{3pt} 
\renewcommand{\arraystretch}{0.95} 
\setlength\LTleft{0pt}
\setlength\LTright{0pt}

\begin{longtable}{lccc||lccc||lccc}
\caption{Classification of self-dual codes of lengths 2, 4, 6, and 8 over $\mathbb{F}_{13}\cong \mathbb{Z}[i]/(2+3i)$ and $\mathbb{F}_{17}\cong \mathbb{Z}[i]/(1+4i)$}
\label{fig:code_results_13}\\
\toprule
Code & $|\mbox{Aut}(\mathcal{C})|$ & $d_H$ & $d_\pi$ & Code & $|\mbox{Aut}(\mathcal{C})|$ & $d_H$ & $d_\pi$ & Code & $|\mbox{Aut}(\mathcal{C})|$ & $d_H$ & $d_\pi$\\
\midrule
\endfirsthead
\hline
Code & $|\mbox{Aut}(\mathcal{C})|$ & $d_H$ & $d_\pi$ & Code & $|\mbox{Aut}(\mathcal{C})|$ & $d_H$ & $d_\pi$ & Code & $|\mbox{Aut}(\mathcal{C})|$ & $d_H$ & $d_\pi$\\
\hline
\endhead
$C_{13, 2}(*)$ & 4 & 2 & 2 & $C_{17, 2}(*)$ & 4 & 2 & 2 & $C_{17, 8, 21}$ & 8 & 4 & 4\\
$C_{13, 4, 1}$ & 32 & 2 & 2 & $C_{17, 4, 1}$ & 32 & 2 & 2 & $C_{17, 8, 22}$ & 8 & 4 & 4\\
$C_{13, 4, 2}(*)$ & 24 & 3 & 5 & $C_{17, 4, 2}(*)$ & 16 & 3 & 5 & $C_{17, 8, 23}$ & 24 & 4 & 4\\
$C_{13, 6, 1}$ & 384 & 2 & 2 & $C_{17, 6, 1}$ & 384 & 2 & 2 & $C_{17, 8, 24}$ & 24 & 4 & 4\\
$C_{13, 6, 2}(*)$ & 36 & 3 & 5 & $C_{17, 6, 2}$ & 64 & 2 & 2 & $C_{17, 8, 25}$ & 32 & 4 & 6\\
$C_{13, 6, 3}$ & 96 & 2 & 2 & $C_{17, 6, 3}$ & 48 & 4 & 4 & $C_{17, 8, 26}(*)$ & 8 & 4 & 7\\
$C_{13, 6, 4}$ & 24 & 4 & 4 & $C_{17, 6, 4}$ & 24 & 4 & 4 & $C_{17, 8, 27}$ & 4 & 4 & 6\\
$C_{13, 6, 5}$ & 48 & 4 & 4 & $C_{17, 6, 5}$ & 16 & 3 & 4 & $C_{17, 8, 28}$ & 4 & 4 & 6\\
$C_{13, 8, 1}$ & 6144 & 2 & 2 & $C_{17, 6, 6}(*)$ & 12 & 4 & 6 & $C_{17, 8, 29}$ & 4 & 4 & 6\\
$C_{13, 8, 2}$ & 768 & 2 & 2 & $C_{17, 8, 1}$ & 6144 & 2 & 2 & $C_{17, 8, 30}$ & 4 & 4 & 6\\
$C_{13, 8, 3}$ & 144 & 2 & 2 & $C_{17, 8, 2}$ & 512 & 2 & 2 & $C_{17, 8, 31}$ & 4 & 4 & 6\\
$C_{13, 8, 4}$ & 96 & 2 & 2 & $C_{17, 8, 3}$ & 64 & 2 & 2 & $C_{17, 8, 32}$ & 2 & 4 & 6\\
$C_{13, 8, 5}$ & 192 & 2 & 2 & $C_{17, 8, 4}$ & 192 & 2 & 2 & $C_{17, 8, 33}$ & 4 & 4 & 6\\
$C_{13, 8, 6}$ & 1152 & 3 & 5 & $C_{17, 8, 5}$ & 96 & 2 & 2 & $C_{17, 8, 34}$ & 4 & 4 & 6\\
$C_{13, 8, 7}$ & 36 & 3 & 5 & $C_{17, 8, 6}$ & 48 & 2 & 2 & $C_{17, 8, 35}$ & 48 & 4 & 6\\
$C_{13, 8, 8}$ & 12 & 3 & 4 & $C_{17, 8, 7}$ & 512 & 3 & 5 & $C_{17, 8, 36}$ & 16 & 4 & 6\\
$C_{13, 8, 9}$ & 24 & 3 & 4 & $C_{17, 8, 8}$ & 32 & 3 & 4 & $C_{17, 8, 37}$ & 16 & 4 & 6\\
$C_{13, 8, 10}$ & 96 & 4 & 5 & $C_{17, 8, 9}$ & 16 & 3 & 5 & $C_{17, 8, 38}$ & 4 & 4 & 6\\
$C_{13, 8, 11}$ & 24 & 4 & 4 & $C_{17, 8, 10}$ & 16 & 3 & 4 & $C_{17, 8, 39}$ & 16 & 4 & 6\\
$C_{13, 8, 12}$ & 4 & 4 & 5 & $C_{17, 8, 11}$ & 8 & 3 & 4 & $C_{17, 8, 40}$ & 8 & 4 & 6\\
$C_{13, 8, 13}$ & 16 & 4 & 4 & $C_{17, 8, 12}$ & 4 & 3 & 5 & $C_{17, 8, 41}$ & 16 & 4 & 6\\
$C_{13, 8, 14}(*)$ & 16 & 4 & 6 & $C_{17, 8, 13}$ & 384 & 4 & 4 & $C_{17, 8, 42}$ & 8 & 4 & 6\\
$C_{13, 8, 15}(*)$ & 16 & 4 & 6 & $C_{17, 8, 14}$ & 16 & 4 & 4 & $C_{17, 8, 43}$ & 32 & 5 & 6\\
$C_{13, 8, 16}$ & 384 & 4 & 4 & $C_{17, 8, 15}$ & 32 & 4 & 4 & $C_{17, 8, 44}$ & 48 & 5 & 6\\
$C_{13, 8, 17}$ & 32 & 4 & 4 & $C_{17, 8, 16}$ & 4 & 4 & 4 & $C_{17, 8, 45}$ & 12 & 5 & 6\\
$C_{13, 8, 18}$ & 8 & 4 & 4 & $C_{17, 8, 17}$ & 64 & 4 & 4 & $C_{17, 8, 46}(*)$ & 48 & 5 & 7\\
$C_{13, 8, 19}$ & 24 & 4 & 4 & $C_{17, 8, 18}$ & 8 & 4 & 4 & $C_{17, 8, 47}$ & 16 & 5 & 6\\
$C_{13, 8, 20}(*)$ & 8 & 4 & 6 & $C_{17, 8, 19}$ & 24 & 4 & 4 & & & & \\
$C_{13, 8, 21}(*)$ & 48 & 5 & 6 & $C_{17, 8, 20}$ & 8 & 4 & 4 & & & &\\
\bottomrule
\end{longtable}
\endgroup
\end{center}

For lengths $n=10, 12, 14$ over $\mathbb{F}_{13}$ and lengths $n=10, 12$ over $\mathbb{F}_{17}$, we used the self-dual code construction method in~\cite{IEEE-KC-building}. We list those achieving the largest minimum Mannheim distance among the codes we found. For each code's generator matrix $G=[I|A]$, we give the non-identity part $A$.
\begin{itemize}
    \item A $[10, 5, d_H=5, d_\pi=7]$ code over $\mathbb{F}_{13}$ with generator matrix
    \begin{equation} \label{eq-code-10-5-5-7}
    	A=\begin{bmatrix}
    		9 & 0 & 1 & 11 & 11\\
    		0 & 8 & 6 & 11 & 5\\
    		1 & 6 & 1 & 2 & 3\\
    		11 & 11 & 2 & 3 & 2\\
    		11 & 5 & 3 & 2 & 10
    	\end{bmatrix}.
    \end{equation}
    \item A $[12, 6, d_H=5, d_\pi=8]$ code over $\mathbb{F}_{13}$ with generator matrix
    \[
    A=\begin{bmatrix}
        4 & 0 & 0 & 1 & 11 & 11\\
        0 & 2 & 1 & 0 & 6 & 7\\
        0 & 1 & 12 & 6 & 9 & 7\\
        1 & 0 & 6 & 3 & 11 & 12\\
        11 & 6 & 9 & 11 & 12 & 9\\
        11 & 7 & 7 & 12 & 9 & 6
    \end{bmatrix}.
    \]
    \item A $[14, 7, d_H=5, d_\pi=8]$ code over $\mathbb{F}_{13}$ with generator matrix
    \[
    A=\begin{bmatrix}
        7 & 1 & 0 & 0 & 0 & 7 & 2\\
        1 & 9 & 0 & 1 & 7 & 7 & 0\\
        0 & 0 & 7 & 1 & 0 & 11 & 7\\
        0 & 1 & 1 & 3 & 2 & 0 & 6\\
        0 & 7 & 0 & 2 & 6 & 7 & 11\\
        7 & 7 & 11 & 0 & 7 & 9 & 12\\
        2 & 0 & 7 & 6 & 11 & 12 & 3\\
    \end{bmatrix}.
    \]
    \item A $[10, 5, d_H=5, d_\pi=8]$ code over $\mathbb{F}_{17}$ with generator matrix
    \[
    A=\begin{bmatrix}
        3 & 0 & 1 & 5 & 10\\
        0 & 2 & 12 & 6 & 6\\
        1 & 12 & 16 & 10 & 5\\
        5 & 6 & 10 & 12 & 0\\
        10 & 6 & 5 & 0 & 5\\
    \end{bmatrix}.
    \]
    \item A $[12, 6, d_H=5, d_\pi=8]$ code over $\mathbb{F}_{17}$ with generator matrix
    \[
    A=\begin{bmatrix}
        14 & 0 & 0 & 1 & 15 & 11\\
        0 & 15 & 1 & 0 & 10 & 8\\
        0 & 1 & 3 & 10 & 5 & 0\\
        1 & 0 & 10 & 15 & 15 & 14\\
        15 & 10 & 5 & 15 & 16 & 1\\
        11 & 8 & 0 & 14 & 1 & 5\\
    \end{bmatrix}.
    \]
\end{itemize}

We summarize our current results on the minimum Mannheim distance of self-dual codes over $\mathbb{F}_{13}$ and $\mathbb{F}_{17}$ in Table~\ref{self-dual-F13-F17}. In the column labeled by $\pi$, the listed numbers indicate the size of the field $\mathcal{G}_\pi$ over which the corresponding self-dual codes are defined. For each length $n$, the quantity $d_\pi^*(n)$ denotes the upper bound obtained from the algorithm described above. The last column records the current status of $d_\pi^{SD}(n)$.

\begin{table}[H]
	\caption{The best known values of $d_{\pi}^{SD}(n)$ over $\mathbb{F}_{13}$ and $\mathbb{F}_{17}$}
	\label{self-dual-F13-F17}
	\centering
	\begin{tabular}{cccc||cccc}
		\toprule
        $\pi$ & $n$ & $d_\pi^*(n)$ & $d_\pi^{SD}(n)$ & $\pi$ & $n$ & $d_\pi^*(n)$ & $d_\pi^{SD}(n)$
        \\\midrule
        $13$ & $2$ & $2$ & $2$ & $17$ & $2$ & $2$ & $2$
        \\\hline
        $13$ & $4$ & $5$ & $5$ & $17$ & $4$ & $5$ & $5$
        \\\hline
        $13$ & $6$ & $5$ & $5$ & $17$ & $6$ & $6$ & $6$
        \\\hline
        $13$ & $8$ & $7$ & $6$ & $17$ & $8$ & $9$ & $7$
        \\\hline
        $13$ & $10$ & $9$ & $7-9$ & $17$ & $10$ & $10$ & $8-10$
        \\\hline
        $13$ & $12$ & $10$ & $8-10$ & $17$ & $12$ & $12$ & $8-12$
        \\\hline
        $13$ & $14$ & $12$ & $8-12$ & & & & 
        \\\bottomrule
	\end{tabular}
\end{table}

\begin{prop}
    Suppose there exist self-dual codes over $\mathbb{F}_{13}$ with parameters $[10, 5, d_\pi=9]$ and $[14, 7, d_\pi=12]$. Then they must be MDS codes.
\end{prop}
\begin{proof}
Let $\mathbf{c}\in\mathbb{F}_{13}^{10}$ be a vector of Hamming weight $5$. Since
\[
\mathbb{F}_{13}^\times=\{1, 5, 8, 12\}\cup \{2, 10, 3, 11\}\cup \{4, 7, 6, 9\},
\]
at least two nonzero coordinates $c_i$, $c_j$ of $\mathbf{c}$ are contained in the same coset. Then, there is $\alpha\in\mathbb{F}_{13}^\times$ such that $\alpha c_i, \alpha c_j\in \{1, 5, 8, 12\}$ and the Mannheim weight of $\alpha\mathbf{c}$ is $\le 8$. Thus, if a $[10, 5, d_\pi=9]$ self-dual code exists, then it must be an MDS code.

Similarly, if $\mathbf{c}\in\mathbb{F}_{13}^{14}$ is a vector of Hamming weight $7$, then at least three nonzero coordinates are contained in the same coset. Thus, there is $\alpha\in\mathbb{F}_{13}^\times$ such that the Mannheim weight of $\alpha\mathbf{c}$ is $\le 1+1+1+2+2+2+2=11$. This implies that a $[14, 7, d_\pi=12]$ self-dual code, if it exists, is an MDS code.
\end{proof}

\section{Decoding} \label{sec-decoding}

Let $\mathcal{C}$ be an $[n,k,d_{\pi}]$ code over $\mathcal{G}_\pi$ with the parity-check matrix $H$, where $\pi = a+bi \in \mathcal{G}$, $0 < a < b$, $\gcd(a,b) = 1$ and $p = a^2+b^2 \equiv 1 \pmod{4}$.
Similarly to Hamming distance, the code can correct errors of Mannheim weight exactly $\lfloor (d_{\pi} - 1)/2 \rfloor$.
Let $\vec{c} \in \mathcal{C}$ be the transmitted codeword, and $\vec{r} = \vec{c} + \vec{e}$ be the received vector, where $\vec{e}$ has Mannheim weight less than or equal to $\lfloor (d_{\pi} - 1)/2 \rfloor$.

Suppose that codewords from the code $C$ are being sent over a communication channel.
If a word $\vec{r}$ is received, the nearest neighbour decoding rule (or minimum distance decoding rule) will decode $\vec{r}$ to $\vec{c}$ if $d_{\pi}(\vec{r}, \vec{c})$ is minimal among all the codewords in $C$, i.e.,
\[ d_{\pi}(\vec{r}, \vec{c}) = \min_{\vec{x} \in C} d_{\pi}(\vec{x}, \vec{r}). \]

\noindent
\textbf{Syndrome Decoding Algorithm}

Here we give the following steps to construct a syndrome look-up table assuming complete nearest neighbour decoding.
\begin{enumerate}
\item[{(i)}]
    List all the cosets for the code, choose from each coset a word of least Mannheim weight as coset leader $\vec{u}$.
	Similar to \cite[Exercise 4.44]{book-LX}, the word $\vec{u}$ is the unique coset leader of $\vec{u} + \mathcal{C}$ if the Mannheim weight of $\vec{u}$ is at most $\lfloor (d_{\pi}(\mathcal{C}) - 1)/2 \rfloor$.

\item[{(ii)}]
    Find a parity-check matrix $H$ for the code and, for each coset leader $\vec{u}$, calculate its syndrome $S(\vec{u}) = \vec{u} \cdot H^T$.
\end{enumerate}

\begin{ex}
	Assume that $p=17$, $n=4$, $\alpha = \beta = 1+i$ and
	\[
	H =
	\begin{pmatrix}
		1 & \beta & \beta^2 & \beta^3 \\
		1 & \beta^5 & \beta^{10} & \beta^{15} \\
	\end{pmatrix}
	=
	\begin{pmatrix}
		1 & 1+i & 2i & -2+2i \\
		1 & -4-4i & -2i & -8+8i \\
	\end{pmatrix}.
	\]
	According to \cite[Example 5.1]{B-Proc}, the code with parity-check matrix $H$ is an $i$-cyclic code, a constant cyclic code.
	Since $p=17$, we have $4^4 \equiv 1 \pmod{17}$.
	That is to say, such a code is actually a code determined the following parity-check matrix
	\[
	H =
	\begin{pmatrix}
		1 & 1+i & 2i & -2+2i \\
		1 & -4-4i & -2i & -8+8i \\
	\end{pmatrix}
	=
	\begin{pmatrix}
		1 & 5 & 8 & 6 \\
		1 & 14 & 9 & 7 \\
	\end{pmatrix}.
	\]
	Let $\vec{r} = (2, -2i, -1-i, 1) = (2, 9, 12, 1)$ be the received vector, whose syndrome is $S(\vec{r}) = \vec{r} \cdot H^T = (13, 5)$.
	By computer search, we find that the coset with syndrome $(13, 5)$ is exactly $(8,1,0,0) + \mathcal{C}$.
	Moreover, the vectors with smallest Mannheim weight are
	\[ \vec{u}_0 = (8,1,0,0), \vec{u}_1 = (0,1,4,13), \vec{u}_2 = (0,9,13,0), \vec{u}_3 = (0,0,5,4), \vec{u}_4 = (1,0,13,13). \]
	The submitted codeword $\vec{s} = \vec{r} - \vec{u}$ should be one of the following vectors:
	\[ \vec{s}_0 = (11,8,12,1), \vec{s}_1 = (2,8,8,5), \vec{s}_2 = (2,0,16,1), \vec{s}_3 = (2,9,7,14), \vec{s}_4 = (1,9,16,5). \]
\end{ex}

\begin{ex}
	Assume that $p = 13$, $n = 10$, and $G = [I_5 \ A]$, where $A$ is defined as \eqref{eq-code-10-5-5-7}.
	Then $G$ can generate a $[10, 5, d_H=5, d_\pi=7]$ code over $\mathbb{F}_{13}$ with parity-check matrix $H = [-A^T \ I_5]$.
	This code can correct errors of Hamming weight $2$ and Mannheim weight $3$.
	\begin{itemize}
		\item Let
		\[ \vec{r} = (1, 2, 0, 1, 11, 2, 1, 9, 12, 8) \]
		be the received vector, whose syndrome is $S(\vec{r}) = \vec{r} \cdot H^T = (4,10,0,6,5)$.
		By computer search, we find that the coset with syndrome $S(\vec{r})$ is exactly $\vec{e} + \mathcal{C}$, where
		\[ \vec{e} = (1, 2, 0, 0, 0, 0, 0, 0, 0, 0). \]
		It is clearly that $\vec{e}$ has Hamming weight $2$ and Mannheim weight $3$.
		Thus the submitted codeword $\vec{s} = \vec{r} - \vec{e}$ should be
		\[ \vec{s} = \vec{r} - \vec{e} = (0, 0, 0, 1, 11, 2, 1, 9, 12, 8). \]
		
		\item If the received vector is
		\[ \vec{r} = (1, 1, 1, 1, 11, 2, 1, 9, 12, 8), \]
		then its syndrome is $S(\vec{r}') = \vec{r}' \cdot H^T$.
		The coset with syndrome $S(\vec{r})$ is exactly $\vec{e} + \mathcal{C}$, where
		\[ \vec{e} = (1, 1, 1, 0, 0, 0, 0, 0, 0, 0) \]
		has Hamming weight $3$ and Mannheim weight $3$.
		That is to say, with respect to Mannheim distance, the submitted codeword should be
		\[ \vec{s} = \vec{r} - \vec{e} = (0, 0, 0, 1, 11, 2, 1, 9, 12, 8). \]
		
		We can check that $\vec{e}$ has the smallest Hamming weight in the coset $\vec{e} + \mathcal{C}$.
		Thus, with respect to Hamming distance, the submitted codeword should be
		\[ \vec{s} = \vec{r} - \vec{e} = (0, 0, 0, 1, 11, 2, 1, 9, 12, 8). \]
	\end{itemize}
\end{ex}

\noindent
\textbf{Syndrome Decoding For Perfect Codes}

Let $\mathcal{C}$ be a perfect code with Mannheim weight $d_{\pi}$.
Then the number of cosets for $\mathcal{C}$ is the same as the number of codewords whose weight is less than or equal to $\lfloor (d_{\pi} - 1)/2 \rfloor$.
Then the decoding can be done in the following steps:
\begin{enumerate}
\item[{(i)}]
    When $\vec{c} \in \mathcal{C}$ is sent and $\vec{w}$ is received, calculate its syndrome $S(\vec{w}) = \vec{w} \cdot H^T$.

\item[{(ii)}]
    If $S(\vec{w}) = \vec{0}$, then assume $\vec{c} = \vec{w}$.
	
\item[{(iii)}]
    If $S(\vec{w}) \neq \vec{0}$, then there exists $\vec{u} \in \mathbb{F}_p^n$ such that $S(\vec{w}) = S(\vec{u})$.
	So assume that $\vec{c} = \vec{w} - \vec{u}$.
\end{enumerate}

\section{Conclusion}
In this paper, we have obtained several bounds for the minimum Mannheim distances of codes over Gaussian integers. By presenting explicit examples of codes that meet these bounds, we have demonstrated that our bounds are tight. Moreover, we have also presented a Mannheim-distance version of the sphere-packing bound. Using this bound, we have identified the parameters for which a 2-error perfect code may exist. Using a MacWilliams-type identity, we have derived bounds for self-dual codes with respect to the Mannheim distance and have characterized several optimal self-dual codes attaining the largest possible Mannheim distance. Finally, by presenting examples where decoding is not possible under the Hamming distance but becomes possible under the Mannheim distance, we have showed that our results are meaningful.

\section*{Declarations}  The authors declare no conflict of interest.


\begin{thebibliography}{00}
    \bibitem{book-AA-QDE} T. Andreescu, D. Andrica, \textit{Quadratic Diophantine Equations}. Springer, New York. 2015.

    \bibitem{DM-BGGHK-sdcode} K. Betsumiya, S. Georgiou, T. A. Gulliver, M. Harada, and C. Koukouvinos, On self-dual codes over some prime fields, Discrete Math., 262(1-3):37-58, 2003.

    \bibitem{IEEE-BR-1perfect} J. Borges and J. Rifà, A Characterization of 1-perfect additive codes, IEEE Trans. Inf. Theory, 45(5):1688-1697, 2002.

    \bibitem{B-Proc} S. Bouyuklieva, Applications of the Gaussian integers in coding theory, Proceedings of the 3rd International Colloquium on Differential Geometry and its Related Fields, Veliko Tarnovo, 2012.

    \bibitem{IEEE-C-Z2k} C. Carlet, $\mathbb{Z}_{2^k}$-linear codes, IEEE Trans. Inf. Theory, 44(4):1543-1547, 1998.

    \bibitem{IC-CW-leechannel} J. C. Y. Chiang and J. K. Wolf, On channels and codes for the Lee metric, Inf. Control, 19(2):159-173, 1971.

    \bibitem{DCC-CK-self-dual} W.-H. Choi, J.-L. Kim, An improved upper bound on self-dual codes over finite fields GF(11), GF(19), and GF(23), Des. Codes Cryptogr., 90:2735–2751, 2022.

    \bibitem{PIT-CH} I. Constantinescu and W. Heise, A metric for codes over residue class rings of integers, Probl. Inf. Transm, 33(3):22–28, 1997.

    \bibitem{IEEE-EY-Leemulti} T. Etzion and E.Yaakobi, Error-correction of multidimensional bursts, IEEE Trans. Inf. Theory, 55(3):961-976, 2009.

    \bibitem{SAM-GW-Leeperfect} S. W. Golomb and L. R. Welch, Perfect codes in the Lee metric and the packing of polyominoes, SIAM J. Appl. Math., 18(2):302-317, 1970.

    \bibitem{IEEE-GS-chainring} M. Greferath and S. E. Schmidt, Gray isometries for finite chain rings and a nonlinear ternary $(36, 3^{12}, 15)$ code, IEEE Trans. Inf. Theory, 45(7):2522-2524, 1999.

    \bibitem{IBM-G-Griesmer} J. H. Griesmer, A bound for error-correcting codes, IBM J. Res. Dev., 4(5):532–542, 1960.

    \bibitem{IEEE-HKCSS-Z4} A. R. Hammons, P. V. Kumar, A. R. Calderbank, N. J. A. Sloane and P. Sole, The $\mathbb{Z}_4$-linearity of Kerdock, Preparata, Goethals, and related codes, IEEE Trans. Inf. Theory, 40(2):301-319, 1994.

    \bibitem{FFA-H-self-dual} M. Harada, The existence of a self-dual [70,35,12] code and formally self-dual codes, Finite Fields Their Appl., 3(2):131–139, 1997.

    \bibitem{FFA-H-self-dual-2} M. Harada, Construction of extremal Type II $\mathbb{Z}_{2k}$-codes, Finite Fields Their Appl., 87:102154, 2023.

    \bibitem{IEEE-H-gaussian} K. Huber, Codes over Gaussian integers, IEEE Trans. Inf. Theory, 40(1):207-216, 2002.

    \bibitem{AAECC-H-MTTDMM} K. Huber, The MacWilliams theorem for two-dimensional modulo metrics, Appl. Algebra Engrg. Comm. Comput.,  8(1):41-48, 1997.

    \bibitem{IEEE-K-building} J.-L. Kim, New extremal self-dual codes of lengths 36, 38 and 58, IEEE Trans. Inf. Theory, 47(1):386–393, 2001.

    \bibitem{IEEE-KC-building} J.-L. Kim and W.-H. Choi, Self-dual codes, symmetric matrices, and eigenvectors, IEEE Access, 9:104294-104303, 2021.

    \bibitem{JCTA-LL-self-dual} H. Lee, Y. Lee, Construction of self-dual codes over finite rings $\mathbb Z_{p^m}$, J. Combin. Theory Ser. A, 115:407–422, 2008.

    \bibitem{book-LX} S. Ling, C. Xing, \textit{Coding Theory: A First Course}. Cambridge University Press, New York. 2004.

    \bibitem{JLMS-L-QDE-4} W. Ljunggren, Some remarks on the Diophantine equations $x^2 - D y^4 = 1$ and $x^4 - D y^2 = 1$, J. Lond. Math. Soc. 1(1):542-544, 1966.

	\bibitem{TIT-MBG-metrics} C. Martínez, R. Beivide and E. Gabidulin, Perfect Codes for Metrics Induced by
	Circulant Graphs, IEEE Trans. Inf. Theory, 53(9):3042-3052, 2007.

    \bibitem{IEEE-M-CM} H. Matsui, An algorithm for finding self-orthogonal and self-dual codes over Gaussian and Eisenstein integer residue rings Via Chinese Remainder Theorem, IEEE Access, 11:23260-23267, 2023.
	
	\bibitem{book-ME-problem} M. R. Murty, J. Esmonde, \textit{Problems in Algebraic Number Theory}. second edition, Springer, New York. 2005.

    \bibitem{QIC-OG-quantum} M. Özen and M. Güzeltepe, Quantum codes from codes over Gaussian integers with respect to the Mannheim metric, Quantum Info. Comput. 12(9–10):813-819, 2012.

    \bibitem{ITI-P-plotkin} M. Plotkin, Binary codes with specified minimum distance, IRE Trans. Inf. Theory, 6(4):445–450, 1960.

    \bibitem{IEEE-RS-LeeBCH} R. M. Roth and P. H. Siegel, Lee-metric BCH codes and their application to constrained and partial-response channels, IEEE Trans. Inf. Theory, 40(4):1083-1096, 1994.

    \bibitem{IEEE-SSAA-GaussianBCH} M. Sajjad, T. Shah, M. Alammari and H. Alsaud, Construction and decoding of BCH-codes over the Gaussian field, IEEE Access, 11:71972-71980, 2023.

    \bibitem{IEEE-S-LeeRM} K.-U. Schmidt, Complementary sets, generalized Reed-Muller codes, and power control for OFDM, IEEE Trans. Inf. Theory, 53(2):808-814, 2007.

    \bibitem{FCN-SLLK-QAM} J. Ssimbwa, B. Lim, J.-H. Lee and Y.-C. Ko, A survey on robust modulation requirements for the next generation personal satellite communications, Front. Commun. Netw., 3:850781, 2022.

    \bibitem{SAM-T-fieldperfect} A. Tietäväinen, On the nonexistence of perfect codes over finite fields, SIAM J. Appl. Math., 24(1):88-96, 1973.

    \bibitem{IEEE-WLZ} R. Wan, Y. Li, S. Zhu, New MDS self-dual codes over finite field $\mathbb F_{r^2}$, IEEE Trans. Inf. Theory, 69(8):5009-5016, 2023.
\end{thebibliography}
\end{document}